\documentclass[dvipsnames]{article}

\usepackage{amsmath}
\usepackage{amssymb}
\usepackage{amsthm}

\newtheorem{remark}{Remark}
\newtheorem{example}{Example}
\newtheorem{proposition}{Proposition}
\newtheorem{theorem}{Theorem}
\newtheorem{lemma}{Lemma}

\usepackage{bm}

\usepackage{xspace}
\usepackage{turnstile}

\usepackage{booktabs}
\usepackage{multirow}
\usepackage{pifont}
\usepackage{mathtools}

\usepackage{multicol}
\usepackage{stackengine}
\usepackage{pgf}

\usepackage{tikz}
\usetikzlibrary{decorations.pathreplacing,positioning}
\usetikzlibrary{arrows,automata,calc}
\tikzset{>=stealth, shorten >=1pt}
\tikzset{every edge/.style = {thick, ->, draw}}
\tikzset{every loop/.style = {thick, ->, draw}}
\pgfdeclarelayer{background}
\pgfsetlayers{background,main}

\usetikzlibrary{
  calc,
  arrows,
  positioning,
  shapes,
  shadows,
  automata,
  external,
  shadows,
  decorations.pathreplacing,
  decorations.pathmorphing,
  decorations.markings,
  backgrounds,
  patterns
}

\tikzset{  
  node distance=4em,
  >=stealth',
  shadow/.style		= {opacity=.25, black, shadow xshift=0.08, shadow yshift=-0.08},
  plainnode/.style 	= {draw, ultra thick, fill=gray!10},
  pl0/.style	       = {circle, minimum size=6mm, inner sep=0mm,plainnode, drop shadow = {shadow}},
  ind/.style         = {circle,draw,fill=black},
  nind/.style        = {circle,draw},
  wt/.style          = {text = white, yshift = -0.1em,scale=0.9}
}

\newcommand{\nats}{\mathbb{N}}
\renewcommand{\epsilon}{\varepsilon}
\renewcommand{\phi}{\varphi}
\newcommand{\size}[1]{|#1|}
\newcommand{\pow}[1]{2^{#1}}
\newcommand{\cceq}{\mathop{::=}}
\newcommand{\set}[1]{\{#1\}}
\newcommand{\F}{\mathop{\mathbf{F}\vphantom{a}}\nolimits}
\newcommand{\G}{\mathop{\mathbf{G}\vphantom{a}}\nolimits}
\DeclareMathOperator{\U}{\mathbf{U}}
\newcommand{\X}{\mathop{\mathbf{X}\vphantom{a}}\nolimits}

\newcommand{\ltl}{{LTL}\xspace}
\newcommand{\ctl}{{CTL}\xspace}
\newcommand{\forallctl}{{$\forall$CTL}\xspace}
\newcommand{\ctlstar}{{CTL$^*$}\xspace}
\newcommand{\hyltl}{{HyperLTL}\xspace}
\newcommand{\hyctlstar}{{HyperCTL$^*$}\xspace}

\newcommand{\suffix}[2]{#1[#2,\infty)}
\newcommand{\var}{\mathcal{V}}
\newcommand{\ap}[0]{\mathrm{AP}}

\newcommand{\threeexp}{\textsc{3ExpTime}\xspace}

\newcommand{\tower}{\textsc{Tower}\xspace}

\newcommand{\myquot}[1]{``#1''}
\newcommand{\flip}[1]{\overline{#1}}

\newcommand{\tsys}{\mathfrak{T}}
\newcommand{\traces}{\mathrm{Tr}}

\newcommand{\aut}{\mathcal{A}}
\newcommand{\autp}{\mathcal{P}}
\newcommand{\initmark}{I}

\newcommand{\col}{\Omega}

\newcommand{\combine}[1]{\mathrm{mrg}(#1)}

\newcommand{\ext}{\mathit{ext}}

\newcommand{\vectordots}[2]{
\begin{pmatrix}
  #1 \\
    \vspace{-.5cm}\\
  \vdots \\
    \vspace{-.5cm}\\
  #2\\
\end{pmatrix}}

\newcommand{\vectordotswithtwoargsatstart}[3]{
\begin{pmatrix}
  #1 \\
  #2\\
  \vdots \\
  #3\\
\end{pmatrix}}

\newcommand{\vectoroflengthtwo}[2]{\begin{pmatrix}
  #1 \\
  #2\\
\end{pmatrix}}
\newcommand{\vectoroflengththree}[3]{\begin{pmatrix}
  #1 \\
  #2\\
  #3 \\
\end{pmatrix}}

\renewcommand{\vectordots}[2]{( #1 , \ldots , #2 )}

\renewcommand{\vectordotswithtwoargsatstart}[3]{( #1 , #2 , \ldots , #3 )}

\renewcommand{\vectoroflengthtwo}[2]{( #1 , #2 )}

\renewcommand{\vectoroflengththree}[3]{( #1 , #2 , #3 )}

\newcommand{\game}{\mathcal{G}}
\newcommand{\hasrun}[5]{#1\colon #2 \xrightarrow{#3} #5}
\newcommand{\hasrunacc}[5]{#1\colon #2 \xRightarrow{#3} #5}

\newcommand{\tr}[1]{
\def\a{t}\def\b{#1}
\ifx\a\b{\bm\tilde{t}}\else\bm\tilde{#1}\fi
}

\newcommand{\bddft}{bounded-delay transducer\xspace}
\newcommand{\bddfts}{bounded-delay transducers\xspace}

\newcommand{\transd}{\mathcal{T}}
\newcommand{\transdfunc}[1]{f_{#1}}

\newcommand{\tm}{\mathcal{M}}
\newcommand{\dom}[1]{\mathrm{dom}(#1)}

\newcommand{\prefs}[1]{\mathrm{Prfs}(#1)}

\newcommand{\mchyper}{\textsc{MCHyper}\xspace}
\newcommand{\hypervis}{\textsc{HyperVis\xspace}}

\newcommand{\SigmaI}{\Sigma_I}
\newcommand{\SigmaO}{\Sigma_O}

\newcommand{\hist}{\mathrm{Hist}}
\newcommand{\upd}{\mathrm{upd}}
\newcommand{\nextmove}{\mathrm{nxt}}

\newcommand{\blocks}{\mathcal{B}}
\newcommand{\block}{b}
\newcommand{\blockseq}{\overline{b}}
\newcommand{\extend}{\mathrm{ext}}
\newcommand{\free}{\mathrm{shft}}
\newcommand{\sink}{s_\bot}
\newcommand{\restrict}[2]{#1\!\upharpoonright\!#2}
\newcommand{\lag}{\Delta}

\newcommand{\moore}{\mathcal{S}}

\newcommand{\owner}{o}

\newcommand{\equivnew}[1]{\equiv_{#1}}

\newcommand{\bigo}{\mathcal{O}}

\newcommand{\ind}{\mathit{idx}}
\newcommand{\autsize}{\mathit{sz}}

\usepackage{fullpage}
\usepackage[color=purple!40]{todonotes}

\title{Tracy, Traces, and Transducers:\\ Computable Counterexamples and Explanations\\ for HyperLTL Model-Checking\footnote{Supported by DIREC - Digital Research Centre Denmark.}}
\author{Sarah Winter (IRIF \& Université Paris Cité, Paris, France)\\ Martin Zimmermann (Aalborg University, Aalborg, Denmark)}
\date{}

\begin{document}

\maketitle

\begin{abstract}
    HyperLTL model-checking enables the automated verification of information-flow properties for se\-cu\-ri\-ty-critical systems. 
However, it only provides a binary answer. 
Here, we consider the problem of computing counterexamples and explanations for HyperLTL model-checking, thereby considerably increasing its usefulness.

Based on the maxim ``counterexamples/explanations are Skolem functions for the existentially quantified trace variables'', we consider (Turing machine) computable Skolem functions. 
As not every finite transition system and formula have computable Skolem functions witnessing that the system satisfies the formula, we consider the problem of deciding whether such functions exist.
Our main result shows that this problem is decidable by reducing it to solving multiplayer games with hierarchical imperfect information. 
Furthermore, our algorithm also computes transducers implementing such functions, if they exist.
\end{abstract}

%%%%%%%%%%%%%%%%%%%%%%%%%%%%%%%%%%%%%%%%%%%%%%%%%%%%%%%%%%%%%%%%%%%%%%
%%%%%%%%%%%%%%%%%%%%%%%%%%%%%%%%%%%%%%%%%%%%%%%%%%%%%%%%%%%%%%%%%%%%%%
%%%%%%%%%%%%%%%%%%%%%%%%%%%%%%%%%%%%%%%%%%%%%%%%%%%%%%%%%%%%%%%%%%%%%%
\section{Introduction}
\label{sec:intro}

\subparagraph*{Prologue.}
Tracy sits in her office and needs to print her latest travel reimbursement claim. After hitting the print button, she walks to the printer room only to find out that the document has not been printed. So, she walks back to her office, hits the print button again, walks to the printer and is slightly surprised to find her document. Sometimes Tracy wonders whether the print system is nondeterministic. If only there was a way to find out.
\medskip

Information-flow properties, which are crucial in the specification of security-critical systems,
require the simultaneous reasoning about multiple executions of a system.
However, most classical specification languages like \ltl and \ctlstar refer to a single execution trace at a time. 
Clarkson and Schneider~\cite{ClarksonS10} coined the term \emph{hyperproperties} for properties that require the reasoning about multiple traces.
Just like ordinary trace and branching-time properties, hyperproperties can be specified using temporal logics, e.g., \hyltl and \hyctlstar~\cite{ClarksonFKMRS14}, expressive, but intuitive specification languages that are able to express typical information-flow properties such as noninterference, noninference, declassification, and input determinism.
Due to their practical relevance and theoretical elegance, hyperproperties and their specification languages have received considerable attention during the last decade.

\hyltl is obtained by extending \ltl~\cite{Pnueli77}, the most influential specification language for linear-time properties, by trace quantifiers to refer to multiple executions of a system. 
Hence, a \hyltl formula is indeed evaluated over a set of traces, which forms the universe for the quantifiers.
For example, the \hyltl~formula
$
\phi_{\mathrm{id}}=\forall \pi, \pi'.\ \G( i_\pi \leftrightarrow i_{\pi'}) \rightarrow \G (o_\pi \leftrightarrow o_{\pi'})
$
expresses input determinism, i.e., every pair of traces that always has the same input (represented by the proposition~$i$) also always has the same output (represented by the proposition~$o$).
Having learned about \hyltl, Tracy wonders whether she can formally prove that the print system violates  $\phi_{\mathrm{id}}$.

In this work, we focus on the model-checking problem for \hyltl, which intuitively asks whether a given (finite model of a) system satisfies a given \hyltl specification.
This problem is decidable, albeit \tower-complete~\cite{Rabe16Diss,MZ20}.

But the model-checking problem as described above is \myquot{just} a decision problem, i.e., the user only learns whether the system satisfies the specification or not, but not the \emph{reason} it does or does not.
It has been argued that this binary answer is in general not useful~\cite{K}: Most real-life systems are too complex to be modelled faithfully by a finite transition system. 
Hence, one always checks an abstraction, not the actual system.
Then, a positive answer to the model-checking problem does not show that the actual system is correct, bugs in it might have been abstracted away when constructing a finite transition system modelling it.
The actual killer application of model-checking is the automated generation of counterexamples in case the specification is not satisfied by the abstraction.
Given a counterexample in the abstraction one can then check whether this (erroneous) behaviour also exists in the actual system, or whether it was introduced during the abstraction~\cite{cegar}.
In the latter case, the abstraction has to be refined and checked again.
But if the erroneous behaviour can be found in the actual system, then this bug can be fixed in the actual system. 

But what is a counterexample in \hyltl model-checking?
For the formula~$\phi_{\mathrm{id}}$ expressing input determinism this is straightforward: if a transition system does not satisfy the formula, then it has two traces that coincide on their input, but not on their output. 
However, the situation becomes more interesting in the presence of existentially quantified variables and quantifier alternations.
Consider, for example, a formula of the form~$\phi = \exists \pi \forall \pi'.\ \psi$ with quantifier-free $\psi$ and let $\tsys$ be a transition system with set~$\traces(\tsys)$ of traces.
If $\tsys \not\models \phi$, then for every choice of $t \in \traces(\tsys)$ there is a $t' \in \traces(\tsys)$ such that the variable assignment~$\set{\pi\mapsto t, \pi'\mapsto t'} $ does not satisfy $\psi$. 
Thus, a counterexample is described by a Skolem function~$f\colon \traces(\tsys) \rightarrow \traces(\tsys)$ for the existentially quantified variable~$\pi'$ in the negation~$\forall \pi \exists\pi'.\ \neg \psi$ of $\phi$.
It gives, for every choice~$t$ for the existentially quantified~$\pi$ in $\phi$ a trace~$f(t)$ for the universally quantified~$\pi'$ in $\phi$ such that $\set{\pi\mapsto t, \pi'\mapsto f(t)} \models \neg \psi$, i.e., $\set{\pi\mapsto t, \pi'\mapsto f(t)} \not\models \psi$, thereby explaining for every choice of $t$ why it is not a good one.
The maxim \myquot{counterexamples are Skolem functions for existentially quantified variables in the negation of the specification} is true for arbitrary formulas. 
But before we explore this approach further, let us first consider a second application.

Explainability, the need to explain to, e.g., users, customers, and regulators, what a system does, is an aspect of system design that gains more and more significance.
This is in particular true when it comes to systems designed by algorithms, e.g., machine-learning or synthesis.
For any nontrivial system of this kind, it is impossible for humans to develop an explanation of their behaviour or a witness for their correctness.
This is a major obstacle preventing the wide-spread use of (unexplained) machine-generated software in safety-critical applications~\cite{xai}.
Also here, \hyltl model-checking can be useful: Assuming the system is supposed to satisfy a \hyltl specification and indeed does so, then Skolem functions \myquot{explain} why the specification is satisfied.

\subparagraph*{Our Contributions.} 
In this work, we are interested in computing counterexamples/ex\-pla\-na\-tions for \hyltl, which, as argued above, boils down to computing Skolem functions for \hyltl. 
Before we explain our contributions, let us remark that counterexamples are just explanations for the negation of the specification, as we have seen above. 
Hence, in the following we will focus on explanations, as this setting spares us from dealing with a negation.
Also, let us remark that for every transition system~$\tsys$ and every \hyltl formula~$\phi$, we either have $\tsys \models \varphi$ or $\tsys \models \neg \varphi$.
Hence, our framework will either explain why $\tsys$ satisfies $\phi$ or explain why $\tsys$ satisfies $\neg \phi$, i.e., explain why $\tsys$ does not satisfy $\phi$.

In general, we are given a transition system~$\tsys$ and a \hyltl formula~$\phi$ such that $\tsys \models \phi$, and we want to compute Skolem functions for the existentially quantified variables in $\phi$.
Note that the actual explanation-phase employing the Skolem functions is an interactive process between the user (i.e., Tracy) and these functions: Tracy has to specify choices for the universally quantified variables, which are then fed into the Skolem functions, yielding choices for the existentially quantified variables such that the combination of all of these traces satisfies the quantifier-free part of the specification.

To \emph{apply} Skolem functions in that manner, they need to be finitely representable. 
To be as general as possible, we consider here functions that are computable by Turing machines (in a very natural sense). 
But even for such a general model, $\tsys \models \varphi$ may not have a computable explanation. The underlying reason is that such Turing machines can only compute continuous functions:
Intuitively, if two inputs coincide on a \myquot{long} prefix, then the corresponding outputs also coincide on a \myquot{long} prefix.
However, it is straightforward to construct a pair~$\tsys \models \varphi$ that does not have continuous Skolem functions (see Theorem~\ref{thm_nocomp}).
Hence, our main focus is on the following question: given $\tsys$ and $\varphi$ with $\tsys \models \varphi$, is $\tsys \models \varphi$ witnessed by computable Skolem functions?
Our main result shows that this problem is decidable. 
To prove it, we combine techniques developed in the theory of uniformization~\cite{FWjournal}, delay games~\cite{KleinZimmermann}, and multiplayer games with hierarchical imperfect information~\cite{bbb} to express the existence of computable Skolem functions by a multi-player game with hierarchical imperfect information. 
Intuitively, there is one player for each existentially quantified variable and they form a coalition against a player corresponding to the universally quantified variables. Hierarchical imperfect information then captures the structure of the quantifier prefix, e.g., the Skolem function for $\pi_1$ in a formula of the form~$\forall \pi_0 \exists \pi_1 \forall \pi_2 \exists \pi_3.\ \psi $ depends only on $\pi_0$ while the one for $\pi_3$ depends on $\pi_0$ and $\pi_2$ (as usual one can assume that Skolem functions only depend on universally quantified variables).
Furthermore, delay games are a general approach to deciding the existence of continuous functions in synthesis and uniformization~\cite{HL72,HKT,KleinZimmermann,FWjournal}.

As a byproduct of our game-theoretic characterization, we show that if $\tsys \models \varphi$ has (Turing machine) computable Skolem functions, then it also has ones that are computed by word-to-word (one-way) transducers with bounded delay between input and output, a much more modest machine model.
In fact, our algorithm computes transducers implementing Skolem functions whenever computable Skolem functions exist.
This allows for the effective computation and simulation of computable Skolem functions as described above.

%%%%%%%%%%%%%%%%%%%%%%%%%%%%%%%%%%%%%%%%%%%%%%%%%%%%%%%%%%%%%%%%%%%%%%
%%%%%%%%%%%%%%%%%%%%%%%%%%%%%%%%%%%%%%%%%%%%%%%%%%%%%%%%%%%%%%%%%%%%%%
%%%%%%%%%%%%%%%%%%%%%%%%%%%%%%%%%%%%%%%%%%%%%%%%%%%%%%%%%%%%%%%%%%%%%%
\section[Preliminaries]{Preliminaries\protect\footnotemark}\footnotetext{We like the twentieth letter of the alphabet. In fact, we like it so much that we use $t$ to denote traces, $T$ to denote sets of traces, $\tsys$ to denote transition systems, and $\transd$ to denote transducers. We hope this footnote will help the reader keep track of them.}
\label{sec:prels}

We denote the set of nonnegative integers by $\nats$. 
The domain of a partial function~$f \colon A \rightarrow B$ is denoted by $\dom{f} = \set{a \in A \mid f(a) \text{ is defined}}$.
More generally, we denote the domain~$\set{a \in A \mid (a,b)\in R \text{ for some } b\in B}$ of a relation~$R \subseteq A \times B$ by $\dom{R}$.

\subsection{Languages, Transition Systems, and Automata}
An alphabet is a nonempty finite set. 
The sets of finite and infinite words over an alphabet~$\Sigma$ are denoted by $\Sigma^*$ and $\Sigma^\omega$, respectively. The length of a finite word~$w$ is denoted by $\size{w}$. 
Given $n$ infinite words~$w_0,\ldots, w_{n-1}$, let their merge (also known as zip), which is an infinite word over $\Sigma^n$, be defined as
  %ACTA UNCOMMENT\\ \begin{multline*}
\[
    \combine{w_0, \ldots, w_{n-1}}  =
    %ACTA UNCOMMENT\\
    \vectordots{w_0(0)}{w_{n-1}(0)}\vectordots{w_0(1)}{w_{n-1}(1)}\vectordots{w_0(2)}{w_{n-1}(2)}\cdots .
\]
  %ACTA UNCOMMENT\\ \end{multline*}
We define $\combine{w_0, \ldots, w_{n-1}}$ for finite words~$w_0, \ldots, w_n$ of the same length analogously.

The set of prefixes of an infinite word~$w = w(0) w(1) w(2) \cdots \in \Sigma^\omega$ is $\prefs{w} = \set{w(0) \cdots w(i-1) \mid i\ge 0}$, which is lifted to languages~$L \subseteq \Sigma^\omega$ via $\prefs{L} = \bigcup_{w \in L} \prefs{w}$.
A language~$L\subseteq \Sigma^\omega$ is closed if $\set{w \in \Sigma^\omega \mid \prefs{w} \subseteq \prefs{L}} \subseteq L$.

Throughout this paper, we fix a finite set~$\ap$ of atomic propositions. 
A transition system~$\tsys = (V,E,v_\initmark, \lambda)$ consists of a finite set~$V$ of vertices, a set~$E \subseteq V \times V$ of (directed) edges, an initial vertex~$v_\initmark \in V$, and a labelling~$\lambda\colon V \rightarrow \pow{\ap}$ of the vertices by sets of atomic propositions.
We assume that every vertex has at least one outgoing edge.
A path~$\rho$ through~$\tsys$ is an infinite sequence~$\rho = v_0v_1v_2\cdots$ of vertices with  $v_0 = v_\initmark$ and $(v_n,v_{n+1})\in E$ for every $n \ge 0$.
The trace of $\rho$ is defined as $ \lambda(\rho ) = \lambda(v_0)\lambda(v_1)\lambda(v_2)\cdots \in (\pow{\ap})^\omega$.
The set of traces of $\tsys$ is $\traces(\tsys) = \set{\lambda(\rho) \mid \rho \text{ is a path of $\tsys$}}$.

\begin{remark}
\label{rem_closed}
The following facts follow directly from the definition of closed languages.
\begin{enumerate}
    \item\label{rem_closed_tsys} Let $\tsys$ be a transition system. Then, $\traces(\tsys)$ is closed.
    \item\label{rem_closed_product} If $L_0, \ldots,L_{n-1} \subseteq \Sigma^\omega$ are closed, then so is \mbox{$
    \set{\combine{w_0, \ldots, w_{n-1}} \mid w_i \in L_i \text{ for } 0\le i< n}
    .$}
\end{enumerate}
\end{remark}

A Büchi automaton~$\aut = (Q, \Sigma, q_\initmark, \delta, F)$ consists of a finite set~$Q$ of states containing the initial state~$q_\initmark \in Q$ and the subset~$F \subseteq Q$ of accepting states, an alphabet~$\Sigma$, and a transition function~$\delta \colon Q \times \Sigma \rightarrow \pow{Q}$.
Let $w = w(0) w(1) w(2) \cdots \in \Sigma^\omega$.
A run of $\aut$ on $w$ is a sequence~$q_0 q_1 q_2 \cdots $ with $q_0 = q_\initmark$ and $q_{n+1} \in \delta(q_n, w(n))$ for all $n\ge 0$.
A run~$q_0q_1q_2 \cdots$ is (Büchi) accepting if there are infinitely many~$n \in \nats$ with $q_n \in F$.
The language (Büchi) recognized by $\aut$, denoted by $L(\aut)$, is the set of infinite words over $\Sigma$ that have an accepting run of $\aut$ on $w$ that is accepting.

Büchi automata recognize exactly the $\omega$-regular languages, but require nondeterminism to do so.
In the following, it is sometimes prudent to work with deterministic $\omega$-automata.
A deterministic parity automaton~$\autp = (Q, \Sigma, q_\initmark, \delta, \col)$ consists of a finite set~$Q$ of states containing the initial state~$q_\initmark \in Q$, an alphabet~$\Sigma$, a transition function~$\delta \colon Q \times \Sigma \rightarrow Q$, and a coloring~$\col\colon Q \rightarrow \nats$ of the states by colors in $\nats$.
Let $w = w(0) w(1) w(2) \cdots \in \Sigma^\omega$.
Then, $\autp$ has a unique run~$q_0 q_1 q_2 \cdots $ on~$w$, defined as $q_0 = q_\initmark$ and $q_{n+1} = \delta(q_n, w(n))$ for all $n\ge 0$.
A run~$q_0q_1q_2 \cdots$ is (parity) accepting if the maximal color appearing infinitely often in the sequence~$\col(q_0)\col(q_1)\col(q_2)\cdots $ is even.
The language (parity) recognized by $\aut$, denoted by $L(\aut)$, is the set of infinite words over $\Sigma$ such that the run of $\aut$ on $w$ is accepting.

%%%%%%%%%%%%%%%%%%%%%%%%%%%%%%%%%%%%%%%
%%%%%%%%%%%%%%%%%%%%%%%%%%%%%%%%%%%%%%%
%%%%%%%%%%%%%%%%%%%%%%%%%%%%%%%%%%%%%%%
\subsection{\hyltl}
\label{subsec_hyperltl}

The formulas of \hyltl are given by the grammar
\begin{align*}
\phi  {}& \cceq {}  \exists \pi.\ \phi \mid \forall \pi.\ \phi \mid \psi \qquad\quad \psi   \cceq {}  a_\pi \mid \neg \psi \mid \psi \vee \psi \mid \X \psi \mid \psi \U \psi    
\end{align*}
where $a$ ranges over $\ap$ and where $\pi$ ranges over a fixed countable set~$\var$ of {(trace) variables}. Conjunction~($\wedge$), exclusive disjunction~($\oplus$), implication~($\rightarrow$), and equivalence~($\leftrightarrow$) are defined as usual, and the temporal operators \myquot{eventually}~($\F$) and \myquot{always}~($\G$) are derived as $\F\psi = \neg \psi\U \psi$ and $\G \psi = \neg \F \neg \psi$. A {sentence} is a  formula without free variables, which are defined as expected.

The semantics of \hyltl is defined with respect to a {trace assignment}, a partial mapping~$\Pi \colon \var \rightarrow (\pow{\ap})^\omega$. The assignment with empty domain is denoted by $\Pi_\emptyset$. Given a trace assignment~$\Pi$, a variable~$\pi$, and a trace~$t$ we denote by $\Pi[\pi \rightarrow t]$ the assignment that coincides with $\Pi$ everywhere but at $\pi$, which is mapped to $t$. 
Furthermore, $\suffix{\Pi}{j}$ denotes the trace assignment mapping every $\pi$ in $\Pi$'s domain to $\Pi(\pi)(j)\Pi(\pi)(j+1)\Pi(\pi)(j+2) \cdots $, the suffix of $\Pi(\pi)$ starting at position $j$.

For sets~$T$ of traces and trace assignments~$\Pi$ we define 
\begin{itemize}
	\item $(T, \Pi) \models a_\pi$ if $a \in \Pi(\pi)(0)$,
	\item $(T, \Pi) \models \neg \psi$ if $(T, \Pi) \not\models \psi$,
	\item $(T, \Pi) \models \psi_1 \vee \psi_2 $ if $(T, \Pi) \models \psi_1$ or $(T, \Pi) \models \psi_2$,
	\item $(T, \Pi) \models \X \psi$ if $(T,\suffix{\Pi}{1}) \models \psi$,
	\item $(T, \Pi) \models \psi_1 \U \psi_2$ if there is a $j \ge 0$ such that $(T,\suffix{\Pi}{j}) \models \psi_2$ and for all $0 \le j' < j$: $(T,\suffix{\Pi}{j'}) \models \psi_1$, 
	\item $(T, \Pi) \models \exists \pi.\ \phi$ if there exists a trace~$t \in T$ such that $(T,\Pi[\pi \rightarrow t]) \models \phi$, and 
	\item $(T, \Pi) \models \forall \pi.\ \phi$ if for all traces~$t \in T$: $(T,\Pi[\pi \rightarrow t]) \models \phi$. 
\end{itemize}
We say that $T$ {satisfies} a sentence~$\phi$ if $(T, \Pi_\emptyset) \models \phi$. In this case, we write $T \models \phi$ and say that $T$ is a {model} of $\phi$. 
A transition system~$\tsys$ satisfies $\phi$, written $\tsys \models \phi$, if $\traces(\tsys)\models \phi$. 
% Two \hyltl sentences~$\varphi$ and $\varphi'$ are equivalent if $T \models \varphi$ iff $T \models \varphi'$ for every set~$T$ of traces.
Although \hyltl sentences are required to be in prenex normal form, they are closed under Boolean combinations, which can be easily seen by transforming such a formula into an equivalent formula in prenex normal form. 
In particular, the negation~$\neg \phi$ of a sentence~$\varphi$ satisfies $T \models \neg \varphi$ iff  $T\not\models\varphi$.
Also, note that the statement~$(T,\Pi) \models \psi$ for quantifier-free formulas~$\psi$ is independent of $T$. Hence, we often just write $\Pi \models \psi$ for the sake of readability.

\begin{remark}
\label{remark_automataconstruction}
Let $\psi$ be a quantifier-free \hyltl formula (with $k$ free variables~$\pi_0, \ldots, \pi_{k-1}$) and let $\tsys$ be a transition system.
There is an (effectively constructible) Büchi automaton~$\aut_\psi^\tsys$ such that $L(\aut_\psi^\tsys)$ is equal to 
\[
 \{\combine{\Pi(\pi_0), \ldots, \Pi(\pi_{k-1})} \mid \Pi(\pi_i) \in \traces(\tsys) \text{ for $0 \le i < k-1$ and }  (\traces(\tsys),\Pi) \models \psi\}.
\]
It can be obtained by noting that $\psi$ is almost an \ltl formula as it is quantifier-free, but its atomic propositions are still labeled by trace variables, i.e., they are all of the form~$a_{\pi_i}$ for some $i \in \set{0, \ldots, k-1}$.
Let $\psi'$ be the (proper) \ltl formula obtained from $\psi$ by replacing each atomic proposition~$a_{\pi_j}$ by the atomic proposition~$(a,j)$, i.e., $\psi'$ is defined over the set~$\ap\times\set{0, \ldots, k-1}$ of atomic propositions.
For $\psi'$ there exists a Büchi automaton~$\aut_{\psi'}$ with $\size{\psi}\cdot 2^{\size{\psi}}$ many states recognizing the language
\[
L(\aut_{\psi'}) = \set{ t \in (\pow{\ap\times\set{0, \ldots, k-1}})^\omega \mid t \models \psi'}
\]
(see, e.g.,~\cite{BK08}).
By replacing each letter~$A \in \pow{\ap\times\set{0, \ldots, k-1}}$ on a transition of $\aut_{\psi'}$ by the letter
\[
\vectordots{\set{a \in \ap \mid (a,0) \in A}}{\set{a \in \ap \mid (a,k-1) \in A}} \in \left( \pow{\ap} \right)^k
\]
we obtain a Büchi automaton~$\aut_{\psi}$ recognizing 
\[\set{ \combine{\Pi(\pi_0), \ldots, \Pi(\pi_{k-1})} \mid \Pi(\pi_j) \in (\pow{\ap})^\omega \text{ for all $0 \le j < k$ and }   (\traces(\tsys),\Pi) \models \psi }.\]
Finally, by taking the product of $\aut_{\psi}$ with $k$ copies of $\tsys$ (where the $i$-th one restricts the $i$-th trace of $\combine{\Pi(\pi_0), \ldots, \Pi(\pi_{k-1})}$ to traces of $\tsys$), we obtain the desired Büchi automaton~$\aut_\psi^\tsys$ recognizing
\[
\{\combine{\Pi(\pi_0), \ldots, \Pi(\pi_{k-1})} \mid \Pi(\pi_i) \in \traces(\tsys) \text{ for $0 \le i < k-1$ and }  (\traces(\tsys),\Pi) \models \psi\}.
\]
Finally, it has $\size{\psi}\cdot 2^{\size{\psi}}\cdot \size{\tsys}^k$ many states.
\end{remark}

%%%%%%%%%%%%%%%%%%%%%%%%%%%%%%%%%%%%%%%%%%%%%%%%%%%%%
%%%%%%%%%%%%%%%%%%%%%%%%%%%%%%%%%%%%%%%%%%%%%%%%%%%%%
%%%%%%%%%%%%%%%%%%%%%%%%%%%%%%%%%%%%%%%%%%%%%%%%%%%%%
\subsection{Skolem Functions for \hyltl}
Let $\varphi = Q_0 \pi_0 \cdots Q_{k-1} \pi_{k-1}.\ \psi $ be a \hyltl sentence such that $\psi$ is quantifier-free and let $T$ be a set of traces. 
Moreover, let $i \in \set{0,1,\ldots, k-1}$ be such that $Q_i = \exists$ and let $U_i = \set{j < i \mid Q_{j} = \forall}$ be the indices of the universal quantifiers preceding $Q_i$.
Furthermore, let $f_i \colon T^{\size{U_i}} \rightarrow T$ for each such $i$ 
(note that $f_i$ is a constant, if $U_i$ is empty).
We say that a variable assignment~$\Pi$ with $\dom{\Pi} \supseteq \set{\pi_0, \pi_1, \ldots, \pi_{k-1}}$ is consistent with the $f_i$ if
$\Pi(\pi_i) \in T$ for all $i$ with $Q_i = \forall$ and $\Pi(\pi_i) = f_i(\Pi(\pi_{i_0}), \Pi(\pi_{i_1}), \ldots, \Pi(\pi_{i_{{\size{U_i}-1}}}))$ for all $i$ with $Q_i = \exists$, where $U_i = \set{i_0 < i_1 < \cdots < i_{\size{U_i}-1}}$.
If $\Pi \models \psi$ for each $\Pi$ that is consistent with the $f_i$, then we say that the $f_i$ are Skolem functions witnessing $T \models \varphi$.

\begin{remark}
$T \models \varphi$ iff there are Skolem functions for the existentially quantified variables of $\varphi$ that witness $T \models \varphi$.
\end{remark}

Note that only traces for universal variables are inputs for Skolem functions, but not those for existentially quantified variables. 
As usual, this is not a restriction, as the inputs of a Skolem function for an existentially quantified variable~$\pi_i$ is a superset of the inputs of a Skolem function for another existentially quantified variable~$\pi_{j}$ with $j < i$.

\begin{example}
\label{example_reconstruct}
Let $\varphi = \forall \pi \exists \pi_1 \exists \pi_2.\ \G(a_\pi \leftrightarrow (a_{\pi_1} \oplus a_{\pi_2}))$.
We have $(\pow{\set{a}})^\omega \models \varphi$.
Now, for every function~$f_1 \colon (\pow{\set{a}})^\omega \rightarrow (\pow{\set{a}})^\omega$, there is a function~$f_2 \colon (\pow{\set{a}})^\omega \rightarrow (\pow{\set{a}})^\omega$ such that $f_1,f_2$ are Skolem functions witnessing $(\pow{\set{a}})^\omega \models \varphi$, i.e., we need to define $f_2$ such that
$(f_2(t))(n) = (f_1(t))(n)$
for all $n \in\nats$ such that $t(n) = \emptyset$ and
$(f_2(t))(n) = \flip{(f_1(t))(n)} $
for all $n \in\nats$ such that $t(n) = \set{a}$, where $\flip{\set{a}} = \emptyset$ and $\flip{\emptyset} = \set{a}$.
Hence, $f_2$ depends on $f_1$, but the value of $f_1(t)$ (for the existentially quantified~$\pi_1$) does not need to be an input to $f_2$, it can be determined from the input~$t$ for the universally quantified $\pi$. 
This is not surprising, but needs to be taken into account in our constructions.
\end{example}

\section{Problem Statement}
\label{sec_problemstatement}

Our goal is to construct (Turing machine) computable Skolem functions that serve as algorithmic explanations for the satisfaction of a \hyltl property.
At first glance, this is a very ambitious goal, as it requires working with Turing machines processing infinite inputs and producing infinite outputs.
To tackle this issue, we re-formulate our problem as a (rather non-standard, more synthesis-like) uniformization problem.

Here, we consider the variant where one is given a relation~$R \subseteq A \times B$ and the goal is to determine whether there is a function uniformizing $R$ (i.e., a partial function~$f\colon A \rightarrow B$ such that $\dom{f} = \dom{R}$ and $\set{(a,f(a)) \mid a \in \dom{R}} \subseteq R$) that is computed by a machine from some fixed class of machines.
For the case where $A$ is the Cartesian product of the set of traces of a transition system and $B$ is the set of traces of the same transition system, we have captured the problem of computing a Skolem function as a uniformization problem.

Thus, for a sentence with quantifier prefix~$\forall^*\exists^*$ we can obtain computable Skolem functions by interpreting the problem as a uniformization problem.
However, for more complex quantifier prefixes, this is no longer straightforward, as the dependencies between the variables have to be considered, e.g., the Skolem function of an existentially quantified variable only has inputs corresponding to outermore universally quantified variables, but may also depend on outermore existentially quantified variables, i.e., outputs are also (implicit) inputs for other functions.

Another issue is that Turing machines are a very expressive model of computation. 
Filiot and Winter~\cite{FWjournal} studied synthesis of computable functions from rational specifications (e.g., specifications recognized by a Büchi automaton): they proved that uniformization by Turing machines coincides with uniformization by transducers with bounded delay (if the domain of the specification is closed), a much nicer class of machines computing functions from infinite words to infinite words.
Crucially, the functions computed by such transducers are also continuous in the Cantor topology over infinite words (we refer to \cite{FWjournal} for definitions and details). 
In the setting of Skolem functions for \hyltl model-checking for $\forall^*\exists^*$-sentences, that means that if two inputs agree on a \myquot{long} prefix, then the corresponding outputs also agree on a \myquot{long} prefix. 
Continuity is a desirable property when using Skolem functions on-the-fly: 
If Tracy has fixed a long prefix of the inputs for the Skolem functions, then future inputs do not change the output prefixes already produced by the Skolem functions for the fixed prefix.

However, we will show that there is a \hyltl sentence that does not have continuous Skolem functions (and thus also no computable ones).
Hence, it is natural to ask if it is decidable whether a given pair~$(\tsys,\phi)$ has a computable explanation. 
We prove that this is indeed the case, even for sentences with arbitrary quantifier prefixes.

%%%%%%%%%%%%%%%%%%%%%%%%%%%%%%%%%%%%%%%
%%%%%%%%%%%%%%%%%%%%%%%%%%%%%%%%%%%%%%%
%%%%%%%%%%%%%%%%%%%%%%%%%%%%%%%%%%%%%%%
\subsection{Uniformization by Computable Functions}
\label{subsec_uniformization}

In the following, for languages over the alphabet~$\Sigma \times \Gamma$ recognized by Büchi automata, i.e., $L \subseteq (\Sigma \times \Gamma)^\omega$ we speak about its induced relation
$R_L = \set{ (x, y) \in \Sigma^\omega \times \Gamma^\omega \mid \combine{x,y} \in L }.$
%we encode relations~$R \subseteq \Sigma^\omega \times \Gamma^\omega$ by languages over the alphabet~$\Sigma \times \Gamma$, i.e., $L \subseteq (\Sigma \times \Gamma)^\omega$ represents the relation
%$R_L = \set{ (x, y) \in \Sigma^\omega \times \Gamma^\omega \mid \combine{x,y} \in L }.$
%This encoding allows us to represent relations by Büchi automata. 
For the sake of readability, we often do not distinguish between (automata-recognizable) languages $L \subseteq (\Sigma \times \Gamma)^\omega$ and induced relations~$R_L\subseteq \Sigma^\omega \times \Gamma^\omega$.
A function~$f \colon \Sigma^\omega \rightarrow \Gamma^\omega$ uniformizes a relation~$R \subseteq \Sigma^\omega \times \Gamma^\omega$ if the domain of $f$ is equal to the domain of $R$ and the graph
$
\set{(w,f(w)) \mid w\in \dom{f}}
$
of $f$ is a subset of $R$.
In our setting, the uniformization problem asks whether, for a given relation, there is a computable function that uniformizes~it.

To define the computability of a function from~$\Sigma^\omega $ to $ \Gamma^\omega$, we consider deterministic three-tape Turing machines~$\tm$ with the following setup (following~\cite{FWjournal}): the first tape is a read-only, one-way tape and contains the input in $\Sigma^\omega$, the second one is a two-way working tape, and the third one is a write-only, one-way tape on which the output in $\Gamma^\omega$ is generated. % letter by letter. 
Formally, we say that $\tm$ computes the partial function~$f \colon \Sigma^\omega \rightarrow \Gamma^\omega$ if, when started with input~$w \in \dom{f}$ on the first tape, $\tm$ produces (in the limit) the output~$f(w)$ on the third tape.
Note that we do not require the Turing machine to check whether its input is in the domain of $f$. We just require it to compute the correct output for those inputs in the domain, it may behave arbitrarily on inputs outside of the domain of $f$.
This is done so that the uniformization function only has to capture the complexity of transforming possible inputs into outputs, but does not have to capture the complexity of checking whether an input is in the domain. 
In our setting, this can be taken care of by deterministic $\omega$-automata that can be effectively computed.

We say that such an $\tm$ has bounded delay, if there is a $d \in\nats$ such that to compute the first $n$ letters of the output only $n+d$ letters of the input are read (i.e., the other cells of the input tape are not visited before $n$ output letters have been generated). 

\begin{lemma}
\label{lem_tmboundeddelay}
Let $f$ be computable and let $\dom{f}$ be closed. Then $f$ is computable by a Turing machine with bounded delay.
%If the domain of $f$ is closed, then $\tm$ can be assumed (w.l.o.g.) to have bounded delay.
\end{lemma}

\begin{proof}
This follows from the fact that the set $\Sigma^\omega$ is a compact space when equipped with the Cantor distance. A closed subset of a compact space is compact (see, e.g., \cite{arkhangel1990basic}). Hence, $\dom{f}$ is a compact space.
Further, every computable function is continuous (see, e.g.,~\cite{FWjournal}). 
Now, the Heine-Cantor theorem states that every continuous function between metric spaces $f\colon M \to N$ where $M$ is a compact space is in fact uniformly continuous. 
Thus, $f$ lies in the intersection of the classes of computable and uniformly continuous functions which implies that $f$ can also be computed with bounded delay:

Intuitively, $f$ being uniformly continuous means there exists some $k$ such that in order to get $i$ output symbols it suffices to consider $i+k$ input symbols.
Assume a Turing machine $\tm$ computes $f$.
We briefly sketch how to obtain a Turing machine $\tm'$ from $\tm$ that computes $f$ and has delay at most $k$.
%In principle, $\tm'$ works like $\tm$. 

As long as $\tm$ maintains that the $i$-th output symbol is produced before the $(i+k+1)$-th input symbol is read, $\tm'$ behaves like $\tm$.
However, as soon as $\tm$ would violate this, $\tm'$ continues to simulate $\tm$ on a fixed valid continuation of the input word (regardless of how the actual input word is continued).
We refer to this as dummy continuation.
We note here that our input words are traces of some transition system. 
Hence, an input word is a trace and valid continuations of some trace prefix can easily be generated from the transition system.
After $\tm$ has produced another output symbol (while reading the dummy continuation), $\tm'$ produces the same output symbol.
The choice of the dummy continuation has no relevance for the produced output symbol, as $f$ is uniformly continuous: 
Towards a contradiction, assume the output symbol changes for different dummy inputs.
That means there exist $\alpha$ and $\alpha'$ in the domain of $f$ that agree on a prefix of length $i+k$, but the $i$-th symbol of $f(\alpha)$ is different from the $i$-th symbol of $f(\alpha')$ which contradicts that $f$ is uniformly continuous.

After the output symbol has been produced (using the dummy continuation), the dummy continuation can be discarded and $\tm'$ can restart the simulation of $\tm$ on the actual input in the configuration that $\tm$ was in before the simulation on the dummy continuation started.
It is important to note that the next output symbol that is computed by $\tm$ (on the actual continuation) does not have to be produced by $\tm'$ as it was already determined using the dummy continuation (whereas the concrete dummy continuation does not influence the symbol as argued above).
Hence, $\tm'$ must keep track of the size of the lead of the output symbols obtained on dummy continuations compared to the number of output symbols computed on the actual input word as not to produce outputs multiple times.
\end{proof}

%%%%%%%%%%%%%%%%%%%%%%%%%%%%%%%%%%%%%%%%%%%%%%%%%%%%%%%%%%%%%%%%%%%%%%%%%%%%%%%%%%%%%%%%%%%%%%%%%%%%%%%%%%%%%%%%%%%%%%%%
%%%%%%%%%%%%%%%%%%%%%%%%%%%%%%%%%%%%%%%%%%%%%%%%%%%%%%%%%%%%%%%%%%%%%%%%%%%%%%%%%%%%%%%%%%%%%%%%%%%%%%%%%%%%%%%%%%%%%%%%
%%%%%%%%%%%%%%%%%%%%%%%%%%%%%%%%%%%%%%%%%%%%%%%%%%%%%%%%%%%%%%%%%%%%%%%%%%%%%%%%%%%%%%%%%%%%%%%%%%%%%%%%%%%%%%%%%%%%%%%%
\subsection{Transducers}
\label{subsec_transducers}

Of course, Turing machines are a very expressive model of computation. 
Filiot and Winter show that for the uniformization of $\omega$-regular relations, much less expressiveness is sufficient, i.e., for such relations, transducers, i.e., finite automata with output, suffice.

Formally, a (one-way deterministic finite) transducer~$\transd$ is a tuple~$(Q, \Sigma, \Gamma, q_\initmark, \delta, \col)$ that consists of a finite set~$Q$ of states containing the initial state~$q_\initmark$, an input alphabet~$\Sigma$, an output alphabet~$\Gamma$, a transition function~$\delta\colon Q \times \Sigma \rightarrow Q\times\Gamma^*$, and a coloring~$\col\colon Q\rightarrow \nats$.
The (unique) run of $\transd$ on an input~$w = w(0)w(1)w(2) \cdots \in\Sigma^\omega $ is the sequence~$q_0 q_1 q_2 \cdots$ of states defined by $q_0 = q_\initmark$ and $q_{i+1}$ being the unique state with $\delta(q_i,w(i)) = (q_{i+1},x_{i})$ for some $x_i \in \Gamma^*$.
The run is accepting if the maximal color appearing infinitely often in $\col(q_0)\col(q_1)\col(q_2)\cdots$ is even. 
With the run~$q_0q_1q_2 \cdots$ on $w$ we associate the output~$x_0x_1x_2\cdots$, where the $x_i$ are as defined above.
As the transducer is deterministic, it induces a map from inputs to outputs.
Note that the output may, a priori, be a finite or an infinite word over~$\Gamma$.
In the following, we only consider transducers where the output is infinite for every input with an accepting run.
In this case, $\transd$ computes a partial function~$\transdfunc{\transd} \colon \Sigma^\omega \rightarrow \Gamma^\omega$ defined as follows: the domain of $\transdfunc{\transd}$ is the set of infinite words~$w \in \Sigma^\omega$ such that the run of $\transd$ on $w$ is accepting and $\transdfunc{\transd}(w)$ is the output induced by this (unique) run. 

We say that $\transd$ has delay~$d \in\nats$ if for every accepting run and every induced sequence~$x_0 x_1 x_2 \cdots$ of outputs ($x_i$ is the output on the $i$-th transition), we have $i -d \le \size{x_0 \cdots x_{i-1}} \le i$ for all $i \ge 0$, i.e., the output is, at any moment during an accepting run, at most~$d$ letters shorter than the input and never longer.
We say that $\transd$ is a bounded-delay transducer if there is a $d$ such that it has delay~$d$.

% \begin{remark}
% If $f$ is computed by a 1DFT~$\transd$, then $\dom{f}$ is recognized by a deterministic parity automaton that can be effectively computed from $\transd$.
% \end{remark}

\begin{proposition}[\cite{FWjournal}]
\label{prop_uniformizationcharacterization}
The following are equivalent for a relation~$R$ encoded by a Büchi automaton~$\aut$ and with closed~$\dom{R}$:
\begin{enumerate}
    \item $R$ is uniformized by a computable function.
    \item $R$ is uniformized by a function implemented by a \bddft.
\end{enumerate}
\end{proposition}

As explained above, this covers the case of $\forall^*\exists^*$ formulas. 
In the remainder, we generalize this result to full \hyltl, i.e., arbitrary quantifier alternations.

%%%%%%%%%%%%%%%%%%%%%%%%%%%%%%%%%%%%%%%
%%%%%%%%%%%%%%%%%%%%%%%%%%%%%%%%%%%%%%%
%%%%%%%%%%%%%%%%%%%%%%%%%%%%%%%%%%%%%%%
\section{Computing Skolem Functions for \hyltl}
\label{sec_computable}

Our goal is to determine under which circumstances $\tsys \models\varphi$ has a computable explanation, i.e., there are computable Skolem functions witnessing $\tsys \models\varphi$, and whether such Skolem functions can be computed by \myquot{simpler} models of computation, i.e., \bddfts. 

We start by showing that $\tsys \models\varphi$ does not necessarily have a computable explanation.  

\begin{theorem}
\label{thm_nocomp}
There is a \hyltl sentence~$\phi$ and a transition system~$\tsys$ such that $\tsys \models \phi$ is not witnessed by computable Skolem~functions.
\end{theorem}

\begin{proof}
Consider $\phi = \forall \pi\exists\pi'.\  (\F a_{\pi}) \leftrightarrow (\X a_{\pi'})$ and $\tsys$ with $\traces(\tsys) = \emptyset(\pow{\set{a}})^\omega$.

% \begin{figure}
%     \centering

%         \begin{tikzpicture}[ultra thick]
%             \node[pl0] (1) at (0,0) {\small$\emptyset$};
%             \node[pl0] (2) at (3,0) {\small$\set{a}$};

%             \path[->] 
%             (.85,0) edge (1)
%             (1) edge[bend left=18] (2)
%             (2) edge[bend left=18] (1) 
%             (1) edge[loop left] ()
%             (2) edge[loop right] ()
%             ;
%         \end{tikzpicture}
    
%     \caption{The transition system for the proof of Theorem~\ref{thm_nocomp}.}
%     \label{fig_tsysfortheorem}
% \end{figure}

Towards a contradiction, assume there is a computable Skolem function for $\pi'$.
Then, due to Lemma~\ref{lem_tmboundeddelay}, there is also one that is implemented by a bounded-delay Turing machine~$\tm$, say with delay~$d$.
Now, let $\tm$ run on an input with prefix~$\emptyset^{d+2} \in \prefs{\traces(\tsys)}$.
As $\tm$ has bounded delay, it will produce the first two output letters~$\emptyset A \in\emptyset\pow{\set{a}}$ after processing the prefix~$\emptyset^{d+2}$ (note that all traces of $\tsys$ start with $\emptyset$, the label of the initial state).

If $A = \emptyset$, then the output of $\tm$ on the input~$\emptyset^{d+2}\set{a}^\omega$ starts with $\emptyset\emptyset$ (as this output only depends on the prefix~$\emptyset^{d+2}$), but the input contains an $\set{a}$. These traces do not satisfy~$(\F a_{\pi}) \leftrightarrow (\X a_{\pi'})$.
On the other hand, if $A = \set{a}$, then the output of $\tm$ on the input~$\emptyset^\omega$ starts with $\emptyset\set{a}$ (again, the output only depends on the prefix~$\emptyset^{d+2}$), but the input contains no $\set{a}$. Again, these traces do not satisfy~$(\F a_{\pi}) \leftrightarrow (\X a_{\pi'})$.
So, in both cases, $\tm$ does not implement a Skolem function for $\pi'$, i.e., we have the desired contradiction.
\end{proof}

So, as not every $\tsys\models\phi$ is witnessed by computable Skolem functions, it is natural to ask whether it is decidable, given $\tsys$ and $\phi$, if $\tsys\models\phi$ has such a witness.
Before we study this problem, we consider another example showing that for some transition system~$\tsys$ and sentence~$\phi$, even if $\tsys \models\phi$ \emph{does} have computable Skolem functions, not every (computable) Skolem function is a \myquot{good} Skolem function: Fixing a Skolem function for an outermore variable may block innermore variables having computable Skolem functions.
%the procedure described above may not succeed for some transition system~$\tsys$ and sentence~$\phi$, even if $\tsys \models\phi$ \emph{does} have computable Skolem functions: Not every (computable) Skolem function is a \myquot{good} Skolem function.

\begin{example}
\label{example_badchoice}
Consider the sentence
$
\exists \pi \forall\pi'\exists\pi''.\ (\X a_{\pi}) \rightarrow ((\F a_{\pi'}) \leftrightarrow (\X a_{\pi''}))
$
and a transition system~$\tsys$ with $\traces(\tsys) = \emptyset(\pow{\set{a}})^\omega$. Note that every trace of  $\tsys$ starts with $\emptyset$. 
Also, as the quantification of $\pi$ is not in the scope of any other quantifier we can identify Skolem functions for $\pi$ with traces that are assigned to $\pi$.

Now, if we pick a trace~$t$ for $\pi$ with $a \in t(1)$ then there is no computable Skolem function for $\pi''$ (see Theorem \ref{thm_nocomp}). However, if we pick a trace~$t$ for $\pi$ with $a\notin t(1)$ then every function is a Skolem function for $\pi''$, as satisfaction is independent of the choices for $\pi'$ and $\pi''$ in this case.
In particular, $\pi''$ has a computable Skolem function.
\end{example}

Thus, the wrong choice of a (computable) Skolem function for some variable may result in other variables not having computable Skolem functions.
By carefully accounting for the dependencies between the Skolem functions we show that the existence of computable Skolem functions is decidable.

\begin{theorem}
\label{thm:main}
The following problem is decidable: \myquot{Given a transition system~$\tsys$ and a \hyltl sentence~$\phi$ with $\tsys \models \phi$, is $\tsys \models \phi$ witnessed by computable Skolem functions?}
If the answer is yes, our algorithm computes \bddfts implementing such Skolem functions.
\end{theorem}

The next section is dedicated to presenting a game-theoretic characterization of the existence of computable Skolem functions. 

\section{A Game for Computable Skolem Functions}
\label{sec_abstractgame}

Recall that in Subsection~\ref{subsec_uniformization}, we have explained that the special case of $\forall^*\exists^*$ sentences can be solved by a reduction to a uniformization problem. 
We begin this section by giving some intuition for this reduction. 
To simplify our notation, we consider a sentence of the form~$\forall \pi \exists \pi' .\ \psi$ where $\psi$ is quantifier-free.
Here, we need to decide whether there is a computable function~$f\colon \traces(\tsys) \rightarrow \traces(\tsys)$ such that $\set{\pi\mapsto t, \pi'\mapsto f(t)} \models\psi$ for all $t$.
Note that 
$
\set{\combine{t,t'} \mid t,t' \in \traces(\tsys) \text{ and }  \set{\pi\mapsto t, \pi'\mapsto t'} \models\psi}
$
is accepted by a Büchi automaton (see Remark~\ref{remark_automataconstruction}). 
Hence, the problem indeed boils down to a uniformization problem for an $\omega$-regular relation.
This problem was first posed (and partially solved) by Hosch and Landweber in 1971~\cite{HL72} and later completely solved in a series of works~\cite{HKT,KleinZimmermann,FWjournal}.
Let us sketch the main ideas underlying the solution, as we will generalize them in the following.

Let $L \subseteq (\SigmaI \times \SigmaO)^\omega$ be $\omega$-regular. Then,
the existence of a function~$f\colon \SigmaI^\omega \rightarrow \SigmaO^\omega$ that uniformizes $L$ is captured by a (perfect information) two-player game~$\Gamma(L)$ of infinite duration played between Player~$I$ (the input player) and Player~$O$ (the output player) in rounds~$r = 0,1,2,\ldots$ as follows: In each round~$r$, Player~$I$ picks an $a_r \in \SigmaI$ and then Player~$O$ picks a $b_r \in \SigmaO\cup\set{\epsilon}$. Thus, the outcome of a play of $\Gamma(L)$ is an infinite word~$a_0a_1a_2\cdots \in \SigmaI^\omega$ picked by Player~$I$ and a finite or infinite word~$b_0b_1b_2\cdots \in \SigmaO^* \cup \SigmaO^\omega$ picked by Player~$O$.
The outcome is winning for Player~$O$ if $a_0a_1a_2\cdots \in \dom{L}$ implies $\combine{a_0a_1a_2\cdots,b_0b_1b_2\cdots} \in L$ (which requires that $b_0b_1b_2\cdots$ is infinite, i.e., Player~$O$ has to pick infinitely often a letter from $\SigmaO$).
Now, one can show that a winning strategy for Player~$O$ can be turned into a function that uniformizes $L$ and every function uniformizing $L$ can be turned into a winning strategy for Player~$O$.

So far, the uniformizing function may be arbitrary, in particular not computable. Also, the delay between input and output in plays that are consistent with a strategy can be unbounded, e.g., if Player~$O$ picks $\epsilon$ in every second round. A crucial insight is that this is not necessary: for every $\omega$-regular~$L\subseteq (\SigmaI \times \SigmaO)^\omega$ such that Player~$O$ wins $\Gamma(L)$, there is a bound~$\ell$ (that only depends on the size of a (minimal) Büchi automaton accepting~$L$) such that she has a winning strategy that picks $\epsilon$ at most $\ell$ times~\cite{HKT}. 

This insight allows to change the rules of $\Gamma(L)$, giving Player~$O$ the advantage gained by using $\epsilon$ a bounded number of times from the beginning of a play and grouping moves into blocks of letters of a fixed length. How the block length is obtained is explained below.

The block game~$\Gamma_b(L)$ is played in rounds~$r = 0,1,2,\ldots$ as follows: In round~$0$, Player~$I$ picks two blocks~$x_0,x_1 \in \SigmaI^\ell$ and then Player~$O$ picks a block~$y_0 \in \SigmaO^\ell$. Then, in every round~$r >0$, Player~$I$ picks a block~$x_{r+1} \in \SigmaI^\ell$ and then Player~$O$ picks a block~$y_r \in \SigmaO^\ell$. Note that Player~$I$ is one block ahead, as he has to pick two blocks in round~$0$.
This accounts for the delay allowed in the definition of computable functions.
The outcome of a play of $\Gamma_b(L)$ is an infinite word~$x_0x_1x_2\cdots \in \SigmaI^\omega$ picked by Player~$I$ and an infinite word~$y_0y_1y_2\cdots \in \SigmaO^\omega$ picked by Player~$O$.
The outcome is winning for Player~$O$ if $x_0x_1x_2\cdots \in \dom{L}$ implies $\combine{x_0x_1x_2\cdots,y_0y_1y_2\cdots} \in L$.

Now, one can show that $L$ is uniformizable iff Player~$O$ has a winning strategy for $\Gamma_b(L)$. As $\Gamma_b(L)$ is a finite two-player game with $\omega$-regular winning condition, Player~$O$ has a finite-state winning strategy (one implemented by a transducer). Such a finite-state winning strategy can be turned into a computable function uniformizing $L$, as a transducer can be simulated by a Turing machine.
Hence, $\Gamma_b(L)$ does indeed characterize uniformizability of $\omega$-regular relations by computable functions. 

Next, let us give some intuition of how to obtain the bound~$\ell$. 
To this end, let $\aut$ be a Büchi automaton over some alphabet~$\SigmaI \times \SigmaO$ (we first ignore the acceptance condition in this discussion and later hint at how this is taken into account).
As usual, if two finite words~$w_0 \in (\SigmaI \times \SigmaO)^*$ and $w_1 \in (\SigmaI \times \SigmaO)^*$ induce the same state transformations (e.g., for all states~$p$ and $q$, processing $w_0$ from $p$ leads $\aut$ to $q$ iff processing $w_1$ from $p$ leads $\aut$ to $q$), then these words are indistinguishable for $\aut$ (again, we are ignoring acceptance for the time being), i.e., one can replace $w_0$ by $w_1$ without changing the possible runs that $\aut$ has.
This indistinguishability is captured by an equivalence relation over $(\SigmaI \times \SigmaO)^*$ with finite index.

However, to capture the interaction described in $\Gamma(L)$ above, we need a more refined approach. Assume Player~$I$ picks a sequence~$x \in \SigmaI^*$ of letters. Then, Player~$O$ will have to \myquot{complete} this block by picking a block~$y \in \SigmaI^{\size{x}}$ so that $\combine{x,y}$ is processed by the automaton.
In this situation, we can say that $x_0$ and $x_1 \in \Sigma^*$ are equivalent if they are indistinguishable w.r.t.\ to their completions to words of the form~$\combine{x_i, y_i}$, e.g., for all states~$p$ and $q$, there is a completion~$\combine{x_0, y_0} \in (\SigmaI\times\SigmaO)^*$ of $x_0$ that leads $\aut$ from $p$ to $q$ iff there is a completion~$\combine{x_1, y_1} \in (\SigmaI\times\SigmaO)^*$ of $x_1$ that leads $\aut$ from $p$ to $q$.
Intuitively, one does not need to distinguish between $x_0$ and $x_1$ because they allow Player~$O$ to achieve the same state transformations in $\aut$.
This indistinguishability is captured by an equivalence relation over $\SigmaI^*$ of finite index. 
Now, $\ell$ can be picked as an upper bound on the length of a minimal word in all equivalence classes. 

Thus, the intuition behind the definition of $\Gamma_b(L)$ is that blocks of length~$\ell$ are \emph{rich} enough to capture the full strategic choices for both players in $\Gamma(L)$: every longer word has an equivalent one of length at most $\ell$. 

Finally, let us briefly mention how to deal with the Büchi acceptance condition we have ignored thus far. As the state transformations are concerned with finite runs of the automaton, we can just keep track of whether an accepting state has been visited or not during this run, all other information is irrelevant. 
Thus, the equivalence relation will take this (single bit of) information into account as well.

After having sketched the special case of a sentence of the form~$\forall \pi \exists \pi' .\ \psi$, let us now illustrate the challenges we have to address to deal with more quantifier alternations, e.g., for a sentence of the form~$\phi = \forall \pi_0 \exists \pi_1 \cdots \forall \pi_{k-2} \exists \pi_{k-1}.\ \psi$.
\begin{itemize}
    \item We will consider a multi-player game with one player being in charge of providing traces for the universally quantified variables (generalizing Player~$I$ above) and one variable player for each existentially quantified variable (generalizing Player~$O$ above), i.e., altogether we have $\frac{k}{2}+1$ players. Thus, the player in charge of the universally quantified variables produces traces~$t_0, t_2, \ldots, t_{k-2}$ while each variable player produces a trace~$t_i$ (one for each odd $i$). These traces are again picked block-wise in rounds.
    
    \item The choices by the variable player producing $t_i$ (i.e., $i$ is odd) may only depend on the traces $t_0, t_1, \ldots, t_{i-1}$ in order to faithfully capture the semantics of $\phi$. Hence, we need to consider a game of imperfect information, which allows us to hide the traces~$t_{i+1}, \ldots, t_{k-1}$ from the player in charge of $\pi_i$.

    \item Recall that Player~$I$ is always one block ahead of Player~$O$ in $\Gamma_b(L)$, which accounts for the delay allowed in the definition of computable functions. With $k$ traces to be picked (and $t_i$ depending on $t_0, t_1, \ldots, t_{i-1}$), there must be a gap of one block for each even $i$.
\end{itemize}

Now, we are able to present the details of our construction.
For the remainder of this section, we fix a \hyltl sentence~$\phi$ and a transition system~$\tsys$ with $\tsys \models \phi$. 
We assume\footnote{The following reasoning can easily be extended to general sentences with arbitrary quantifier prefixes, albeit at the cost of more complex notation.}
$
\phi = \forall \pi_0 \exists \pi_1 \cdots \forall \pi_{k-2} \exists \pi_{k-1}.\ \psi,
$
and use the Büchi automaton~$\aut_\psi^\tsys = (Q, (\pow{\ap})^k, q_{\initmark}, \delta, F)$ constructed in Remark~\ref{remark_automataconstruction} recognizing the language 
\[
\{ \combine{\Pi(\pi_0), \ldots, \Pi(\pi_{k-1})} \mid \Pi(\pi_i) \in \traces(\tsys) \text{ for all $ 0 \le i < k$ and }(\traces(\tsys),\Pi) \models \psi \}.
\]

In the following, we often need to work with tuples of finite words of the same length.
To simplify our notation, from now on we only write $\vectordotswithtwoargsatstart{w_0}{w_1}{w_{i-1}}$ if each $w_j$ is a word in $(\pow{\ap})^*$ such that $\size{w_0} = \size{w_1} = \cdots = \size{w_{i-1}}$.

\subparagraph*{Equivalence Relations.}

We begin by defining equivalence relations that capture the concept of indistinguishability discussed in the intuition above.

We write $\hasrun{\aut}{p}{w}{}{q}$ for states~$p,q$ of a Büchi automaton~$\aut$ over an alphabet~$\Sigma$, if $\aut$ has a run from $p$ to $q$ processing the word~$w\in\Sigma^*$.
Furthermore, we write $\hasrunacc{\aut}{p}{w}{}{q}$, if $\aut$ has a run from $p$ to $q$ processing the word~$w\in\Sigma^*$ such that the run visits at least one accepting state.
Finally, we write $\hasrun{\tsys}{u}{w}{}{v}$ for vertices~$u,v$ of a transition system~$\tsys$, if $\tsys$ has a path from $u$ to $v$ labeled by the word~$w \in (\pow{\ap})^*$.

We continue by defining, for each $1 \le i \le k$, an equivalence relation~$\equivnew{i}$ between $i$-tuples of (finite) words with the intuition that two such tuples are $i$-equivalent if they do not need to be distinguished. For $i = k$, this means that the two tuples cannot be distinguished by $\aut_\psi^\tsys$ while for $1 \le 1 < k$ this means that both $i$-tuples can be completed (by adding an $(i+1)$-th component) so that the resulting $(i+1)$-tuples are $\equivnew{i+1}$-equivalent.

Formally, we define
\[
\vectordotswithtwoargsatstart{w_0}{w_1}{w_{k-1}} 
\equivnew{k} 
\vectordotswithtwoargsatstart{\tr{w}_0}{\tr{w}_1}{\tr{w}_{k-1}}
\]
if 
\begin{itemize}
    \item for all states~$p,q$ of $\aut_\psi^\tsys$ we have $\hasrun{\aut_\psi^\tsys}{p}{\combine{w_0,w_1,\ldots,w_{k-1}}}{}{q}$ if and only if $\hasrun{\aut_\psi^\tsys}{p}{\combine{\tr{w}_0,\tr{w}_1,\ldots,\tr{w}_{k-1}}}{}{q}$, 
    \item for all states~$p,q$ of $\aut_\psi^\tsys$ we have $\hasrunacc{\aut_\psi^\tsys}{p}{\combine{w_0,w_1,\ldots,w_{k-1}}}{}{q}$ if and only if $\hasrunacc{\aut_\psi^\tsys}{p}{\combine{\tr{w}_0,\tr{w}_1,\ldots,\tr{w}_{k-1}}}{}{q}$, 
    and
    \item for all vertices~$u,v$ of $\tsys$ and all $0 \le j \le k-1$ we have $\hasrun{\tsys}{u}{w_j}{}{v}$ if and only if $\hasrun{\tsys}{u}{\tr{w}_j}{}{v}$.
\end{itemize}

\begin{lemma}
\label{lemma_equivkprops}
Let \[w = 
\combine{w_0^0w_0^1w_0^2\cdots,w_1^0w_1^1w_1^2\cdots,\ldots,  w_{k-1}^0w_{k-1}^1w_{k-1}^2\cdots}
\]
and
\[\tr{w} =
\combine{\tr{w}_0^0\tr{w}_0^1\tr{w}_0^2\cdots,\tr{w}_1^0\tr{w}_1^1\tr{w}_1^2\cdots,\ldots,  \tr{w}_{k-1}^0\tr{w}_{k-1}^1\tr{w}_{k-1}^2\cdots}
\]
be such that 
\[
\vectordotswithtwoargsatstart{w_0^n}{w_1^n}{w_{k-1}^n}
\equivnew{k}
\vectordotswithtwoargsatstart{\tr{w}_0^n}{\tr{w}_1^n}{\tr{w}_{k-1}^n}
\]
for all $n$.
\begin{enumerate}
    \item \label{lemma_equivkprops_autk}
        $w \in L(\aut_\psi^\tsys)$ if and only if $\tr{w} \in L(\aut_\psi^\tsys)$.
    \item \label{lemma_equivkprops_tsys}
        For all $0 \le i \le k-1$, we have $w_i^0w_i^1w_i^2\cdots \in \traces(\tsys)$ if and only if $\tr{w}_i^0\tr{w}_i^1\tr{w}_i^2\cdots \in \traces(\tsys)$.
\end{enumerate}
\end{lemma}

\begin{proof}
    \ref{lemma_equivkprops_autk}.) Let $w \in L(\aut_\psi^\tsys)$, 
    Then, there exists an accepting run of $\aut_\psi^\tsys$ on $w$. Let $q_{n+1}$ be the state reached by the run after the prefix
    \[ 
\combine{w_0^0\cdots w_0^{n-1},w_1^0\cdots w_1^{n-1},\ldots,  w_{k-1}^0 \cdots w_{k-1}^{n-1}},
\]
which implies $q_0 = q_\initmark$.
Then, we have $\hasrun{\aut_\psi^\tsys}{q_{n-1}}{\vectordotswithtwoargsatstart{w_0^n}{w_1^n}{w_{k-1}^n}}{}{q_{n}}$ for all $n$ and 
$\hasrunacc{\aut_\psi^\tsys}{q_{n-1}}{\vectordotswithtwoargsatstart{w_0^n}{w_1^n}{w_{k-1}^n}}{}{q_{n}}$ for infinitely many $n$.

Due to 
\[\vectordotswithtwoargsatstart{w_0^n}{w_1^n}{w_{k-1}^n}
\equivnew{k}
\vectordotswithtwoargsatstart{\tr{w}_0^n}{\tr{w}_1^n}{\tr{w}_{k-1}^n}\] for all $n$, we also have 
$\hasrun{\aut_\psi^\tsys}{q_{n-1}}{\vectordotswithtwoargsatstart{\tr{w}_0^n}{\tr{w}_1^n}{\tr{w}_{k-1}^n}}{}{q_{n}}$ for all $n$ and 
$\hasrunacc{\aut_\psi^\tsys}{q_{n-1}}{\vectordotswithtwoargsatstart{\tr{w}_0^n}{\tr{w}_1^n}{\tr{w}_{k-1}^n}}{}{q_{n}}$ for infinitely many $n$.
This allows us to conclude that there is also an accepting run of $\aut_\psi^\tsys$ on $\tr{w}$.

\ref{lemma_equivkprops_tsys}.) The proof is very analogous one to the previous one, we just have to argue about paths and vertices of $\tsys$ instead of runs and states of $\aut_\psi^\tsys$ (and ignore acceptance) and consider the words~$w_i^n$ and $\tr{w}_i^n$ from the $i$-th component instead of full $k$-tuples~$\vectordotswithtwoargsatstart{w_0^n}{w_1^n}{w_{k-1}^n}$ and $\vectordotswithtwoargsatstart{\tr{w}_0^n}{\tr{w}_1^n}{\tr{w}_{k-1}^n}$.
\end{proof}

Now, for $1 \le i < k$, we define $\equivnew{i}$ inductively as follows: 
\[
\vectordotswithtwoargsatstart{w_0}{w_1}{w_{i-1}} 
\equivnew{i} 
\vectordotswithtwoargsatstart{\tr{w}_0}{\tr{w}_1}{\tr{w}_{i-1}}
\]
if 
\begin{itemize}
    \item for all $w_i$ with $\size{w_i} = \size{w_0}$ there exists a $\tr{w}_i$ with $\size{\tr{w}_i} = \size{\tr{w}_0}$ such that \[
\vectordotswithtwoargsatstart{w_0}{w_1}{w_{i}} 
\equivnew{i+1} 
\vectordotswithtwoargsatstart{\tr{w}_0}{\tr{w}_1}{\tr{w}_{i}},
\] and
\item for all $\tr{w}_i$ with $\size{\tr{w}_i} = \size{\tr{w}_0}$ there exists a $w_i$ with $\size{w_i} = \size{w_0}$ such that \[
\vectordotswithtwoargsatstart{w_0}{w_1}{w_{i}} 
\equivnew{i+1} 
\vectordotswithtwoargsatstart{\tr{w}_0}{\tr{w}_1}{\tr{w}_{i}}.
\]
\end{itemize}

\begin{lemma}
\label{lemma_equivifinite}
Every $\equivnew{i}$ is an equivalence relation of finite index.
\end{lemma}

\begin{proof}
By induction over $i$ from $k$ to $1$.
The induction start~$i = k$ was proven by Büchi~\cite{buchi}, so consider $i < k$.

First, it is straightforward to verify that $\equivnew{i}$ is an equivalence relation, as $\equivnew{i+1}$ is an equivalence relation.
Now, we define $\ext(\vectordotswithtwoargsatstart{w_0}{w_1}{w_{i-1}})$ to be the set of $\equivnew{i+1}$-equivalence classes containing a $\vectordotswithtwoargsatstart{w_0}{w_1}{w_{i}}$ for some $w_i$ with $\size{w_i} = \size{w_0}$.
Now, we define $\vectordotswithtwoargsatstart{w_0}{w_1}{w_{i-1}} {\equivnew{i}}' \vectordotswithtwoargsatstart{\tr{w}_0}{\tr{w}_1}{\tr{w}_{i-1}}$ if and only if $\ext(\vectordotswithtwoargsatstart{w_0}{w_1}{w_{i-1}}) = \ext(\vectordotswithtwoargsatstart{\tr{w}_0}{\tr{w}_1}{\tr{w}_{i-1}})$, which is an equivalence relation of finite index: The codomain of $\ext$ has at most $2^n$ elements, where $n$ is the index of $\equivnew{i+1}$.
Finally, ${\equivnew{i}}'$ refines $\equivnew{i}$, which implies that $\equivnew{i}$ has finite index as well. 
\end{proof}

Let $\ell$ be minimal such that each word~$w$ with $\size{w} \ge \ell$ is in an infinite $\equivnew{1}$ equivalence class.
This is well-defined, as $\equivnew{1}$ has finite index, which implies that there are only finitely many words in finite equivalence classes.
A block is a word in~$(\pow{\ap})^\ell$.

Now we have the definitions at hand to define the game~$\game(\tsys, \phi)$ that captures the existence of computable Skolem functions.
To keep the notation manageable, we describe the game abstractly and defer the concrete definition as a multi-player graph game of imperfect information to Section~\ref{sec_conc}.

\subparagraph*{The Abstract Game.}
The game~$\game(\tsys, \phi)$ is played between Player~$U$ who picks traces for the universally quantified variables (by picking blocks) and a coalition of variable players~$\set{1,3, \ldots, k-1}$, who pick traces for the existentially quantified variables (Player~$i$ for $\pi_i$), also by picking blocks.
As in the intuition given above for the case of a formula of the form~$\forall\pi \exists\pi'.\ \psi$, the rules of the game~$\game(\tsys, \phi)$ need to account for the delay inherent to the definition of computable functions. 
In the $\forall\exists$ setting, this is covered by the fact that the player in charge of $\pi$ is one block ahead of the player in charge of $\pi'$. With more quantifier alternations, we generalize this as follows for $\phi = \forall \pi_0 \exists \pi_1 \cdots \forall \pi_{k-2} \exists \pi_{k-1}.\ \psi$:
\begin{itemize}
    \item The player in charge of $\pi_{k-2}$ is one block ahead of the player in charge of $\pi_{k-1}$.
    \item The player in charge of $\pi_{k-3}$ must not be ahead of the player in charge of $\pi_{k-2}$, but may also not be behind.
    \item The player in charge of $\pi_{k-4}$ must be one block ahead of the player in charge of $\pi_{k-3}$. This implies that the player in charge of $\pi_{k-4}$ must be two blocks ahead of the player in charge of $\pi_{k-1}$.
    \item And so on.
\end{itemize}
So, the player in charge of $\pi_{k-1}$ picks one block in round~$0$, the player in charge of $\pi_{k-2}$ picks two blocks in round~$0$ (to be one block ahead), the player in charge of $\pi_{k-3}$ picks two blocks in round~$0$, the player in charge of $\pi_{k-4}$ picks three blocks in round~$0$ and so on. 
In general, we define $\lag_i = \frac{k-(i-1)}{2}$ for (odd) $i \in \set{1,3,\ldots, k-1}$ and $\lag_i = \lag_{i+1}+1$ for (even) $i \in \set{0,2,\ldots, k-2}$, e.g., we have $\lag_{k-1} = 1$, $\lag_{k-2} = 2$, $\lag_{k-3} = 2$, and $\lag_{k-4} = 3$ capturing the \myquot{delay} described above.

Now, we split each round~$r = 0,1,2,\ldots$ into subrounds~$(r,i)$ for by $i = 0,1,\ldots, \mathmbox{k-1}$.
\begin{itemize}
    \item In subround~$(0,i)$ of round~$0$ for even $i$, Player~$U$ picks $\lag_{i}$ blocks~$t_{i-1}^{0}, t_{i-1}^{1}, \ldots, t_{i-1}^{\lag_{i}-1}$.
    
    \item In subround~$(0,i)$ of round~$0$ for odd $i$, Player~$i$ picks $\lag_i$ blocks~$t_i^{0}, t_i^{1}, \ldots, t_i^{\lag_i-1}$.
    
    \item In subround~$(r,i)$ of round~$r>0$ for even $i$, Player~$U$ picks a block~$t_{i-1}^{\lag_i+r}$.
    
    \item In subround~$(r,i)$ of round~$r>0$ for odd $i$, Player~$i$ picks a block~$t_i^{\lag_i-1+r}$.
\end{itemize}
Figure~\ref{fig_playevolution} illustrates the evolution of a play and illustrates the number of blocks picked in round~$0$ and the resulting \myquot{delay} between the selection of blocks for the different variables.

\begin{figure}
    \centering
    \scalebox{.97}{
    \begin{tikzpicture}[ultra thick,yscale=.7]

\def\x{1.25cm}
\def\y{-1.25cm}

%% Player 1

\draw[rounded corners=7] (5.3*\x,-1.6*\y) -- (4.5*\x,-.6*\y) -- (4.5*\x,.5*\y) -- (3.5*\x,.5*\y) -- (3.5*\x,2.5*\y) -- (2.5*\x,2.5*\y) -- (2.5*\x,4.5*\y) -- (1.5*\x,4.5*\y) -- (1.5*\x,5.5*\y);

\draw[rounded corners=7] (6.3*\x,-1.6*\y) -- (5.5*\x,-.6*\y) -- (5.5*\x,.5*\y) -- (4.5*\x,.5*\y) -- (4.5*\x,2.5*\y) -- (3.5*\x,2.5*\y) -- (3.5*\x,4.5*\y) -- (2.5*\x,4.5*\y) -- (2.5*\x,5.5*\y);

\draw[rounded corners=7] (7.3*\x,-1.6*\y) -- (6.5*\x,-.6*\y) -- (6.5*\x,.5*\y) -- (5.5*\x,.5*\y) -- (5.5*\x,2.5*\y) -- (4.5*\x,2.5*\y) -- (4.5*\x,4.5*\y) -- (3.5*\x,4.5*\y) -- (3.5*\x,5.5*\y);

\node[anchor =west,rotate= 44] at (4*\x,-.6*\y) {round $0$};
\node[anchor =west,rotate= 44] at (5*\x,-.6*\y) {round $1$};
\node[anchor =west,rotate= 44] at (6*\x,-.6*\y) {round $2$};
\node[anchor =west,rotate= 44] at (7*\x,-.6*\y) {round $3$};

% round 0
\fill[rounded corners=7, gray!25] (.6*\x, -.4*\y) -- (4.4*\x, -.4*\y) -- (4.4*\x, .4*\y) -- 
(.6*\x,.4*\y) -- cycle;
\fill[rounded corners=7, gray!25] (3.4*\x,.6*\y)  -- (3.4*\x,1.4*\y) -- (.6*\x,1.4*\y) -- (.6*\x,.6*\y) -- cycle;

% round 1
\fill[rounded corners=7, gray!25] (4.6*\x, -.4*\y) -- (5.4*\x, -.4*\y) -- (5.4*\x, .4*\y) -- (4.6*\x, .4*\y) -- cycle;
\fill[rounded corners=7, gray!25] (3.6*\x, .6*\y) -- (4.4*\x, .6*\y) -- (4.4*\x, 1.4*\y) -- (3.6*\x, 1.4*\y) -- cycle;
%\draw[gray!25,line width=10pt] (4*\x, \y) -- (5*\x,0);

% round 2
\fill[rounded corners=7, gray!25] (5.6*\x, -.4*\y) -- (6.4*\x, -.4*\y) -- (6.4*\x, .4*\y) -- (5.6*\x, .4*\y) -- cycle;
\fill[rounded corners=7, gray!25] (4.6*\x, .6*\y) -- (5.4*\x, .6*\y) -- (5.4*\x, 1.4*\y) -- (4.6*\x, 1.4*\y) -- cycle;
%\draw[gray!25,line width=10pt] (5*\x, \y) -- (6*\x,0);

% round 3
\fill[rounded corners=7, gray!25] (6.6*\x, -.4*\y) -- (7.4*\x, -.4*\y) -- (7.4*\x, .4*\y) -- (6.6*\x, .4*\y) -- cycle;
\fill[rounded corners=7, gray!25] (5.6*\x, .6*\y) -- (6.4*\x, .6*\y) -- (6.4*\x, 1.4*\y) -- (5.6*\x, 1.4*\y) -- cycle;
%\draw[gray!25,line width=10pt] (6*\x, \y) -- (7*\x,0);

%% Player 3

% round 0
\fill[rounded corners=7, gray!25] (.6*\x, 1.6*\y) -- (3.4*\x, 1.6*\y) -- (3.4*\x, 2.4*\y) --
(0.6*\x,2.4*\y) -- cycle;
\fill[rounded corners=7, gray!25] (2.4*\x,2.6*\y)  -- (2.4*\x,3.4*\y) -- (.6*\x,3.4*\y) -- (0.6*\x,2.6*\y) -- cycle;

% round 1
\fill[rounded corners=7, gray!25] (3.6*\x, 1.6*\y) -- (4.4*\x, 1.6*\y) -- (4.4*\x, 2.4*\y) -- (3.6*\x, 2.4*\y) -- cycle;
\fill[rounded corners=7, gray!25] (2.6*\x, 2.6*\y) -- (3.4*\x, 2.6*\y) -- (3.4*\x, 3.4*\y) -- (2.6*\x, 3.4*\y) -- cycle;
%\draw[gray!25,line width=10pt] (3*\x, 3*\y) -- (4*\x,2*\y);

% round 2
\fill[rounded corners=7, gray!25] (4.6*\x, 1.6*\y) -- (5.4*\x, 1.6*\y) -- (5.4*\x, 2.4*\y) -- (4.6*\x, 2.4*\y) -- cycle;
\fill[rounded corners=7, gray!25] (3.6*\x, 2.6*\y) -- (4.4*\x, 2.6*\y) -- (4.4*\x, 3.4*\y) -- (3.6*\x, 3.4*\y) -- cycle;
%\draw[gray!25,line width=10pt] (4*\x, 3*\y) -- (5*\x,2*\y);

% round 3
\fill[rounded corners=7, gray!25] (5.6*\x, 1.6*\y) -- (6.4*\x, 1.6*\y) -- (6.4*\x, 2.4*\y) -- (5.6*\x, 2.4*\y) -- cycle;
\fill[rounded corners=7, gray!25] (4.6*\x, 2.6*\y) -- (5.4*\x, 2.6*\y) -- (5.4*\x, 3.4*\y) -- (4.6*\x, 3.4*\y) -- cycle;
%\draw[gray!25,line width=10pt] (5*\x, 3*\y) -- (6*\x,2*\y);

% Player 5
% round 0
\fill[rounded corners=7, gray!25] (.6*\x, 3.6*\y) -- (2.4*\x, 3.6*\y) -- (2.4*\x, 4.4*\y) -- 
(0.6*\x,4.4*\y) -- cycle;
\fill[rounded corners=7, gray!25] (1.4*\x,4.6*\y)  -- (1.4*\x,5.4*\y) -- (.6*\x,5.4*\y) -- (.6*\x,4.6*\y) -- cycle;

% round 1
\fill[rounded corners=7, gray!25] (2.6*\x, 3.6*\y) -- (3.4*\x, 3.6*\y) -- (3.4*\x, 4.4*\y) -- (2.6*\x, 4.4*\y) -- cycle;
\fill[rounded corners=7, gray!25] (1.6*\x, 4.6*\y) -- (2.4*\x, 4.6*\y) -- (2.4*\x, 5.4*\y) -- (1.6*\x, 5.4*\y) -- cycle;
%\draw[gray!25,line width=10pt] (2*\x, 5*\y) -- (3*\x,4*\y);

% round 2
\fill[rounded corners=7, gray!25] (3.6*\x, 3.6*\y) -- (4.4*\x, 3.6*\y) -- (4.4*\x, 4.4*\y) -- (3.6*\x, 4.4*\y) -- cycle;
\fill[rounded corners=7, gray!25] (2.6*\x, 4.6*\y) -- (3.4*\x, 4.6*\y) -- (3.4*\x, 5.4*\y) -- (2.6*\x, 5.4*\y) -- cycle;
%\draw[gray!25,line width=10pt] (3*\x, 5*\y) -- (4*\x,4*\y);

% % round 3
\fill[rounded corners=7, gray!25] (4.6*\x, 3.6*\y) -- (5.4*\x, 3.6*\y) -- (5.4*\x, 4.4*\y) -- (4.6*\x, 4.4*\y) -- cycle;
\fill[rounded corners=7, gray!25] (3.6*\x, 4.6*\y) -- (4.4*\x, 4.6*\y) -- (4.4*\x, 5.4*\y) -- (3.6*\x, 5.4*\y) -- cycle;
%\draw[gray!25,line width=10pt] (4*\x, 5*\y) -- (5*\x,4*\y);

    \node[anchor = east] at (0.7,0) {Player~$U$ (variable~$\pi_0$):};
    \node[anchor = east] at (0.7,\y) {Player~$1$ (variable~$\pi_1$):};
    \node[anchor = east] at (0.7,2*\y) {Player~$U$ (variable~$\pi_2$):};
    \node[anchor = east] at (0.7,3*\y) {Player~$3$ (variable~$\pi_3$):};
    \node[anchor = east] at (0.7,4*\y) {Player~$U$ (variable~$\pi_4$):};
    \node[anchor = east] at (0.7,5*\y) {Player~$5$ (variable~$\pi_5$):};

\node at (7.75*\x,0) {$\cdots$};
\node at (6.75*\x,1*\y) {$\cdots$};
\node at (6.75*\x,2*\y) {$\cdots$};
\node at (5.75*\x,3*\y) {$\cdots$};
\node at (5.75*\x,4*\y) {$\cdots$};
\node at (4.75*\x,5*\y) {$\cdots$};

    \node at (\x,0) {$t_0^0$};
    \node at (\x,\y) {$t_1^0$};
    \node at (\x,2*\y) {$t_2^0$};
    \node at (\x,3*\y) {$t_3^0$};
    \node at (\x,4*\y) {$t_4^0$};
    \node at (\x,5*\y) {$t_5^0$};
    
    \node at (2*\x,0) {$t_0^1$};
    \node at (2*\x,\y) {$t_1^1$};
    \node at (2*\x,2*\y) {$t_2^1$};
    \node at (2*\x,3*\y) {$t_3^1$};
    \node at (2*\x,4*\y) {$t_4^1$};
    \node at (2*\x,5*\y) {$t_5^1$};
    
    \node at (3*\x,0) {$t_0^2$};
    \node at (3*\x,\y) {$t_1^2$};
    \node at (3*\x,2*\y) {$t_2^2$};
    \node at (3*\x,3*\y) {$t_3^2$};
    \node at (3*\x,4*\y) {$t_4^2$};
    \node at (3*\x,5*\y) {$t_5^2$};

    \node at (4*\x,0) {$t_0^3$};
    \node at (4*\x,\y) {$t_1^3$};
    \node at (4*\x,2*\y) {$t_2^3$};
    \node at (4*\x,3*\y) {$t_3^3$};
    \node at (4*\x,4*\y) {$t_4^3$};
    \node at (4*\x,5*\y) {$t_5^3$};

    \node at (5*\x,0) {$t_0^4$};
    \node at (5*\x,\y) {$t_1^4$};
    \node at (5*\x,2*\y) {$t_2^4$};
    \node at (5*\x,3*\y) {$t_3^4$};
    \node at (5*\x,4*\y) {$t_4^4$};
    
    \node at (6*\x,0) {$t_0^5$};
    \node at (6*\x,\y) {$t_1^5$};
    \node at (6*\x,2*\y) {$t_2^5$};
    
    \node at (7*\x,0) {$t_0^6$};

    % \node[rotate=45] at (1.5*\x,0.5*\y) {\scriptsize (0,1)};
    % \node[rotate=45] at (4.5*\x,0.5*\y) {\scriptsize (1,1)};
    % \node[rotate=45] at (5.5*\x,0.5*\y) {\scriptsize (2,1)};
    % \node[rotate=45] at (6.5*\x,0.5*\y) {\scriptsize (3,1)};

    % \node[rotate=45] at (1.5*\x,2.5*\y) {\scriptsize (0,3)};
    % \node[rotate=45] at (3.5*\x,2.5*\y) {\scriptsize (1,3)};
    % \node[rotate=45] at (4.5*\x,2.5*\y) {\scriptsize (2,3)};
    % \node[rotate=45] at (5.5*\x,2.5*\y) {\scriptsize (3,3)};

    % \node[rotate=45] at (1.4*\x,4.4*\y) {\scriptsize (0,5)};
    % \node[rotate=45] at (2.5*\x,4.5*\y) {\scriptsize (1,5)};
    % \node[rotate=45] at (3.5*\x,4.5*\y) {\scriptsize (2,5)};
    % \node[rotate=45] at (4.5*\x,4.5*\y) {\scriptsize (3,5)};

% \node (01) at (-.1*\x-\acm, -1*\y+.5*\acm) {\small subround~$(0,1)$};
% \node (11) at (2.6*\x, -1*\y+.5*\acm) {\small subround~$(1,1)$};
% \node (21) at (5.35*\x+\acm, -1*\y+.5*\acm) {\small subround~$(2,1)$};    

% \node (03) at (-.1*\x-\acm, 4*\y-.5*\acm) {\small subround~$(0,3)$};
% \node (13) at (2.6*\x, 4*\y-.5*\acm) {\small subround~$(1,3)$};
% \node (23) at (5.35*\x+\acm, 4*\y-.5*\acm) {\small subround~$(2,3)$};    
    
% \path[-stealth]
% (1.5*\x,-.2*\y) edge (01)
% (3.85*\x,-.2*\y) edge (11)
% (4.85*\x,-.2*\y) edge[] (21)
% (.85*\x,3.2*\y) edge[] (03)
% (2.15*\x,3.2*\y) edge (13)
% (3.15*\x,3.2*\y) edge (23)
% ;

    \end{tikzpicture}
   } 
    \caption{The evolution of a play of $\game(\tsys, \varphi)$ for a sentence~$\phi$ with six variables. Each gray shape is a subround, consisting of a move of Player~$U$ or a move of one variable player. We have $\lag_5 = 1$, $\lag_4 = \lag_3 = 2$, $\lag_2 = \lag_1= 3$, and $\lag_0 = 4$, which corresponds to the number of blocks picked by the player in charge of variable~$\pi_i$ in round~$0$.}
    \label{fig_playevolution}
\end{figure}
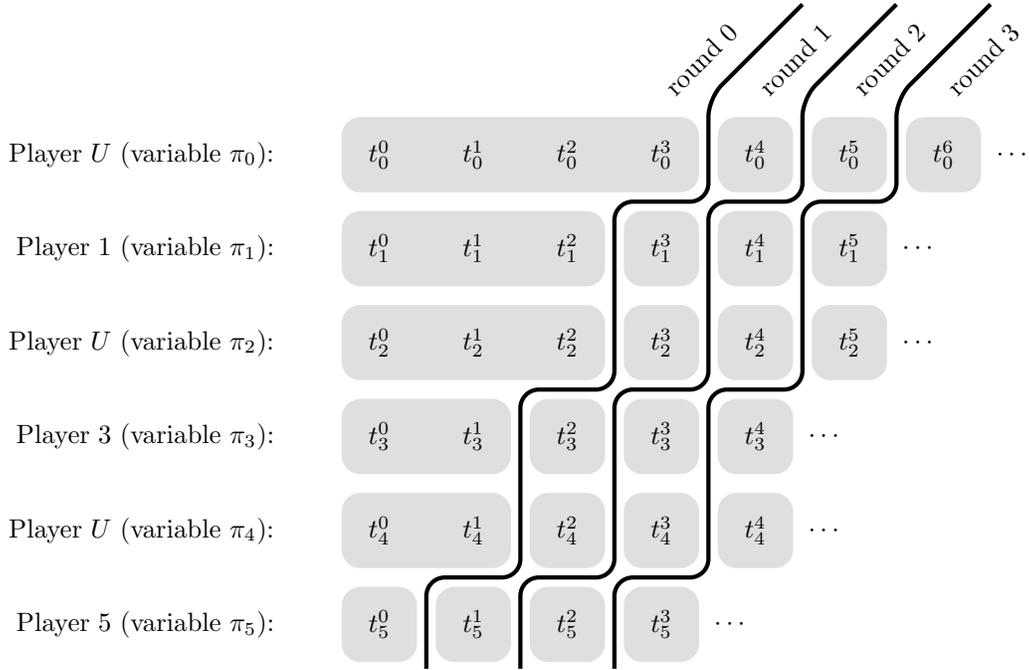

During a play of $\game(\tsys, \varphi)$ the players build traces~$t_0, t_1, \ldots, t_{k-1}$ defined as $t_i = t_i^0t_i^1t_i^2\cdots$.
We call $(t_0, t_1, \ldots, t_{k-1})$ the outcome of the play.
The coalition of variable players wins the play if $t_i \notin \traces(\tsys) $ for some even $i$ or if $\combine{t_0, t_1, \ldots, t_{k-1}} \in L(\aut_\psi^\tsys)$, i.e., the variable assignment mapping each $\pi_i$ to $t_i$ satisfies $\psi$ and each $t_i$ is in $\traces(\tsys)$.

As already alluded to above, the game described above must be a game of imperfect information to capture the fact that the Skolem function for an existentially quantified $\pi_i$ depends only on the universally quantified variables~$\pi_{j}$ with $j \in \set{0,2,\ldots, i-1}$.
Intuitively, we capture this by giving Player~$i$ access to all blocks picked in subrounds~$(r,j)$ with $j \in \set{0,2,\ldots, i-1}$, but hiding all other picks made by the players in subrounds~$(r,j)$ with $j \in \set{1,3,\ldots, i-2, i, i+1, i+2, \ldots, k-1}$. Note that Player~$i$ not having access to their own moves is not a restriction, as they can always be reconstructed, if necessary.

Formally, a strategy for Player~$i$ is a function~$\sigma_i$ mapping sequences of the form\footnote{In this equation, we use column vectors for a more compact presentation.}
{
\renewcommand{\vectordotswithtwoargsatstart}[3]{
\begin{pmatrix}
  #1 \\
  #2\\
  \vdots \\
  #3\\
\end{pmatrix}}
\renewcommand{\vectoroflengthtwo}[2]{\begin{pmatrix}
  #1 \\
  #2\\
\end{pmatrix}}
\begin{align}
\label{eq_stratinput}
\vectordotswithtwoargsatstart{t_0^0}{t_2^0}{t_{i-1}^0}
\vectordotswithtwoargsatstart{t_0^1}{t_2^1}{t_{i-1}^1}
\cdots 
\vectordotswithtwoargsatstart{t_0^{\lag_i+r}}{t_2^{\lag_i+r}}{t_{i-1}^{\lag_i+r}}
\vectordotswithtwoargsatstart{t_0^{\lag_i+r+1}}{t_2^{\lag_i+r+1}}{t_{i-3}^{\lag_i+r+1}}
\vectordotswithtwoargsatstart{t_0^{\lag_i+r+2}}{t_2^{\lag_i+r+1}}{t_{i-5}^{\lag_i+r+2}}
\cdots 
\vectoroflengthtwo{t_{0}^{\lag_i+r+\frac{i}{2}-1}}{t_{2}^{\lag_i+r+\frac{i}{2}-1}}
\left(t_{0}^{\lag_i+r+\frac{i}{2}}\right)
\end{align}
}
for $r \ge 0$ to a block (or a sequence of $\lag_i$ blocks for~$r=0$). The vectors getting shorter at the end is a manifestation the fact that the players in charge of variables~$\pi_j$ with smaller $j$ are  ahead of the players in charge of variables with larger $j'$ (see Figure~\ref{fig_playevolution}).

A (finite or infinite) play is consistent with $\sigma_i$, if the pick of Player~$i$ in each round~$r$ is the one prescribed by $\sigma_i$.
A collection~$(\sigma_i)_{i \in \set{1,3,\ldots, k-1}}$ of strategies, one for each variable player, is winning, if every play that is consistent with all $\sigma_i$ is won by the variable players.
We say that a strategy~$\sigma_i$ is finite-state, if it is implemented by a transducer that reads inputs as in Equation~(\ref{eq_stratinput}) (over some suitable finite alphabet) and produces an output block (or a sequence of $\lag_i$ blocks in round~$0$).

The following lemma shows that the existence of a winning collection of strategies characterizes the existence of computable Skolem functions.
Note that there is a slight mismatch, as the first implication requires the strategies to be finite-state, while the second implication only yields arbitrary strategies. This gap will be closed later.

\begin{lemma}
\label{lemma_correctness}\hfill
\begin{enumerate}
    \item \label{lemma_correctness_fromfsws2compskolem} If the coalition of variable players has a winning collection of finite-state strategies then $\tsys \models \varphi$ has computable Skolem functions.
    \item \label{lemma_correctness_fromcompskolem2ws} If $\tsys \models \varphi$ has computable Skolem functions, then the coalition of variable players has a winning collection of strategies.
\end{enumerate}
\end{lemma}

\begin{proof}
We first show Item~\ref{lemma_correctness_fromfsws2compskolem}. So, let $(\sigma_i)_{i \in \set{1,3,\ldots, k-1}}$ be a winning collection of finite-state strategies for the variable players.
We construct computable Skolem functions~$(f_i)_{i \in \set{1,3,\ldots, k-1}}$ witnessing $\tsys \models \varphi$.
So, fix some $i\in \set{1,3,\ldots, k-1}$.

The machine~$\tm_i$ computing $f_i$ works in iterations $n = 0,1,2,\ldots$ coinciding with the rounds of $\game(\tsys, \phi)$.
Its input is $\combine{t_0, t_2, \ldots, t_{i-1}}$ (encoding $\frac{i}{2}$ input traces as a single infinite word on the input tape),
where we split each $t_{j}$ into blocks~$t_j^0 t_j^1 t_j^2\cdots$.
Recall that the block length~$\ell$ is a constant, i.e., $\tm_i$ can read its input blockwise.

Now, in iteration~$0$, $\tm_i$ reads the first $\lag_i+\frac{i}{2}$ blocks of the input, which yields $\lag_i+\frac{i}{2}$ blocks of each $t_j$. 
These blocks can be used to simulate the moves of Player~$U$ in subrounds~$(0,j)$ for even $j < i$. Note that this does not require all blocks of the $t_j$ for $j > 0$. These have to be stored in the working tape for later use, as the reading tape is one-way.
The simulated moves by Player~$U$ can be fed into the finite-state implementation of $\sigma_i$, yielding blocks $t_i^{0}, t_i^{1}, \ldots, t_i^{\lag_i-1}$ as output.
The word $t_i^{0}t_i^{1}\cdots  t_i^{\lag_i-1}$ is then written to the output tape of $\tm_i$, which completes iteration~$0$.

In general, assume $\tm_i$ has completed iteration $n-1$ and now starts iteration~$n>0$.
This iteration begins with $\tm_i$ reading another block of the input, which yields another block of each $t_j$. 
The new block of $t_0$, and the oldest stored block for each $t_j$ with $j>0$ can be used to continue the simulated play (restricted to moves by Player~$U$ in subrounds for the variables~$\pi_j$ for $j \in \set{0,2,\ldots, i-1}$) by feeding them into the finite-state implementation of $\sigma_i$, yielding a block~$t$ as output. This block is then appended on the output tape. The unused new blocks of $t_j$ with $j>0$ are again stored on the working tape.
This ends iteration~$n$.

To simulate the play, $\tm_i$ can just store the whole play prefix on the working tape. 
To process the play prefix by the finite-state implementation of $\sigma_i$, $\tm_i$ can just store the whole run prefix on the working tape, although a more economical approach is be possible (see the proof of Theorem~\ref{thm:main} on Page~\pageref{page:proofofmainthm}).

Now, we show that the functions~$f_i$ computed by the $\tm_i$ constructed above are indeed Skolem functions witnessing $\tsys \models \varphi$.
To this end, let $\Pi$ with $\dom{\Pi} \supseteq \set{\pi_0, \pi_1, \ldots, \pi_{k-1}}$ be a variable assignment that is consistent with the $f_i$, i.e., each $\Pi(\pi_i)$ with even $i$ is in $\traces(\tsys)$ and each $\Pi(\pi_i)$ with odd $i$ is equal to $f_i(\Pi(\pi_0), \Pi(\pi_2),\ldots,\Pi(\pi_{i-1}))$.
We need to show that each $\Pi(\pi_i)$ for odd $i$ is in $\traces(\tsys)$ (i.e., the functions~$f_i$ are well-defined) and that $\Pi \models \psi$, i.e., $\combine{\Pi(\pi_0), \Pi(\pi_1),\ldots, \Pi(\pi_{k-1}) }\in L(\aut_\psi^\tsys)$. 

By construction, $(\Pi(\pi_0), \Pi(\pi_1),\ldots, \Pi(\pi_{k-1}))$ is the outcome of a play of $\game(\tsys, \phi)$ that is consistent with the $\sigma_i$ and therefore winning for the variable players.
As each $\Pi(\pi_i)$ with even $i$ is in $\traces(\tsys)$, we conclude $\combine{\Pi(\pi_0), \Pi(\pi_1),\ldots, \Pi(\pi_{k-1}) }\in L(\aut_\psi^\tsys)$, as required.
Note that this does also imply $\Pi(\pi_i)$ for odd $i$ is in $\traces(\tsys)$, as $L(\aut_\psi^\tsys)$ only contains tuples of traces from $\traces(\tsys)$.

Now, let us consider Item~\ref{lemma_correctness_fromcompskolem2ws}. 
Let $f_i \colon (\traces(\tsys))^{\frac{i+1}{2}} \rightarrow \traces(\tsys)$ for $i \in \set{1,3,\ldots, k-1}$ be computable Skolem functions witnessing~$\tsys\models\varphi$, say each $f_i$ is implemented by a Turing machine~$\tm_i$.
By Lemma~\ref{lem_tmboundeddelay}, each $\tm_i$ can be assumed to have bounded delay: there is a $d_i$ such that to compute the first $n$ letters of the output only $n+d_i$ letters of the input are read. 
Note that we can run such a Turing machine~$\tm_i$ with delay~$d_i$ on a finite input~$w$ of length~$n+d_i$ and obtain the first $n$ letters of the output $f_i(w')$ of every infinite~$w'$ that starts with the prefix~$w$.
We will apply this fact to simulate the $\tm_i$ on-the-fly on longer and longer prefixes. 

Also note that our definition of a function~$f$ being computed by a Turing machine~$\tm$ only requires it to compute the output~$f(w)$ for all $w \in \dom{f}$, but it can produce arbitrary (even finite) outputs for $w \notin \dom{f}$.
To simplify our construction, we assume here that each $\tm_i$ produces an infinite output for every input, even if it is not in the domain of $f_i$. 
This can be done w.l.o.g., as the $\tm_i$ have bounded delay~$d_i$: as soon as $\tm_i$ wants to access input letter~$n+d_i+1$ without having produced $n+1$ output letters so far (this can be detected, as $d_i$ is a constant), the run does not have delay~$d_i$, which implies that the input cannot be in $\dom{f_i}$. 
Hence, a designated state can be entered, which produces an arbitrary infinite output while ignoring the remaining input.
The resulting machine still has delay~$d_i$, but a complete domain.

Let~$d\in\nats$ be minimal such that each $\tm_i$ has delay at most $d$.
We inductively define a winning collection of strategies for the variable players.

\paragraph*{Round~\boldmath$0$.}

\subparagraph*{Subrounds~$(0,0)$ and $(0,1)$.}
Assume Player~$U$ picks~$t_0^0, t_0^1, \ldots,t_0^{\lag_0-1}$ in subround~$(0,0)$ to start a play.
We fix $\tr{t}_0^0 = t_0^0$ and fix $\tr{t}_0^n$ for $n \in \set{1, 2, \ldots, \lag_0-1}$ such that $\tr{t}_0^n\equivnew{1} t_0^n$ and $\size{\tr{t}_0^0\tr{t}_0^1\cdots\tr{t}_0^n} \ge \size{\tr{t}_0^0\tr{t}_0^1\cdots\tr{t}_0^{n-1}} + d$ for all such $n$. This is always possible, as the $\equivnew{1}$ equivalence class of each $t_0^n$ is infinite and therefore contains arbitrarily long words.

Let $\tr{t}_1^0 \tr{t}_1^1 \cdots \tr{t}_1^{\lag_0-2} $ be the output of $\tm_1$ when given the partial input~$\tr{t}_0^0\tr{t}_0^1\cdots\tr{t}_0^{\lag_0-1}$ such that 
$\size{\tr{t}_1^n} = \size{\tr{t}_0^n}$ for all $n$. This is well-defined by the choice of the length of the $\tr{t}_0^n$ and the fact that $\tm_1$ has delay~$d$. Note that $\tm_1$ might produce even more output on that input. Any such additional output is ignored in this subround.

As we have $t_0^n \equivnew{1} \tr{t}_0^n$ for all such $n$, there also exists a $t_1^n$ with $\size{t_1^n} = \size{t_0^n}$ such that $
\vectoroflengthtwo{t_0^n}{t_1^n} \equivnew{2} \vectoroflengthtwo{\tr{t}_0^n}{\tr{t}_1^n}$.
We define $\sigma_1$ such that it picks $t_1^0, t_1^1, \ldots, t_1^{\lag_0-2}$ in subround~$(0,1)$. As $\lag_0-2 = \lag_1+1$, these are $\lag_1$ many blocks, as required by the definition of $\game(\tsys, \phi)$. 

\subparagraph*{Subrounds~$(0,2)$ and $(0,3)$.}
    Now, assume Player~$U$ picks~$t_2^0, t_2^1, \ldots,t_2^{\lag_2-1}$ in subround~$(0,2)$.
We fix $\tr{t}_2^0 = t_2^0$ and then fix 
 $\tr{t}_2^n$ for $n \in \set{1 ,2, \ldots,  \lag_2-1}$ such that $ 
 \vectoroflengththree{\tr{t}_0^n}{\tr{t}_1^n}{\tr{t}_2^n}
  \equivnew{3} 
 \vectoroflengththree{t_0^n}{t_1^n}{t_2^n}
 $. This is possible, as we have $
\vectoroflengthtwo{t_0^n}{t_1^n} \equivnew{2} \vectoroflengthtwo{\tr{t}_0^n}{\tr{t}_1^n}$ for all such $n$.
Let $\tr{t}_3^0 \tr{t}_3^1 \cdots \tr{t}_3^{\lag_2-2} $ be the output of $\tm_3$ when given the partial input~$\combine{\tr{t}_0^0\tr{t}_0^1\cdots\tr{t}_0^{\lag_2-1}, \tr{t}_2^0\tr{t}_2^1\cdots\tr{t}_2^{\lag_2-1}} $ such that 
$\size{\tr{t}_3^n} = \size{\tr{t}_0^n}$ for all~$n$ (again, this is well-defined due to the choice of the length of the $\tr{t}_0^n$ and $\tm_3$ having delay~$d$, and might require to ignore some output).

As we have $ 
 \vectoroflengththree{\tr{t}_0^n}{\tr{t}_1^n}{\tr{t}_2^n}
  \equivnew{3} 
 \vectoroflengththree{t_0^n}{t_1^n}{t_2^n}
 $ for all such $n$, there also exists a $t_3^n$ with $\size{t_3^n} = \size{t_0^n}$ such that $
\vectordotswithtwoargsatstart{t_0^n}{t_1^n}{t_3^n} \equivnew{4} \vectordotswithtwoargsatstart{\tr{t}_0^n}{\tr{t}_1^n}{\tr{t}_3^n}$.
We define $\sigma_3$ such that it picks $t_3^0 ,t_3^1 ,\ldots, t_3^{\lag_2-2}$ in subround~$(0,3)$.
As $\lag_2-2 = \lag_3 -1$, these are $\lag_3$ many blocks, as required by the definition of $\game(\tsys, \phi)$. 

\subparagraph*{Subrounds~$(0,i-1)$ and $(0,i)$ for odd $i \in \set{5, 7, \ldots, k-1}$.}
Assume Player~$U$ picks
$t_{i-1}^0, t_{i-1}^1, \ldots,t_{i-1}^{\lag_{i-1}-1}$
in subround~$(0,i-1)$.
As before, we fix $\tr{t}_{i-1}^0 = t_{i-1}^0$ and then fix 
 $\tr{t}_{i-1}^n$ for $n \in \set{1,2, \ldots,  \lag_{i-1}-1}$ such that $ 
 \vectordotswithtwoargsatstart{\tr{t}_0^n}{\tr{t}_1^n}{\tr{t}_{i-1}^n}
  \equivnew{i} 
 \vectordotswithtwoargsatstart{t_{0}^n}{t_1^n}{t_{i-1}^n}
 $. This is possible, as $
\vectordotswithtwoargsatstart{t_0^n}{t_1^n}{t_{i-2}^n} \equivnew{i-1} \vectordotswithtwoargsatstart{\tr{t}_0^n}{\tr{t}_1^n}{\tr{t}_{i-2}^n}$ for all such $n$ is an invariant of our construction.

Let $\tr{t}_i^0 \tr{t}_i^1 \cdots \tr{t}_i^{\lag_{i-1}-2}$ be the output of $\tm_i$ when given the partial input
\[\combine{
\tr{t}_0^0\tr{t}_0^1\cdots\tr{t}_0^{\lag_{i-1}-1},\tr{t}_2^0\tr{t}_2^1\cdots\tr{t}_2^{\lag_{i-1}-1},  \ldots,
\tr{t}_{i-1}^0\tr{t}_{i-1}^1\cdots\tr{t}_{i-1}^{\lag_{i-1}-1}} \] such that 
$\size{\tr{t}_i^n} = \size{\tr{t}_0^n}$ for all $n$. 
As in the previous cases, this is well-defined.

As we have $ \vectordotswithtwoargsatstart{\tr{t}_0^n}{\tr{t}_1^n}{\tr{t}_{i-1}^n}
  \equivnew{i}
 \vectordotswithtwoargsatstart{t_{0}^n}{t_{1}^n}{t_{i-1}^n}
 $ for all such $n$, there also exists a $t_i^n$ with $\size{t_i^n} = \size{t_0^n}$ such that $\vectordotswithtwoargsatstart{t_0^n}{t_1^n}{t_i^n} \equivnew{i+1}
 \vectordotswithtwoargsatstart{\tr{t}_0^n}{\tr{t}_1^n}{\tr{t}_i^n}$ for all $n$, satisfying the invariant again.
We define $\sigma_i$ such that it picks $t_i^0, t_i^1,\ldots, t_i^{\lag_{i-1}-2} $ in subround~$(0,i)$.
As $\lag_{i-1}-2 = \lag_i -1$, these are $\lag_i$ many blocks, as required by the definition of $\game(\tsys, \phi)$. 

\paragraph*{Round~\boldmath$r>0$.}
Now, we consider a round~$r> 0$, assuming the $\sigma_i$ are already defined for all earlier rounds. 
The construction is very similar to the one for round~$0$, but simpler as each player (also Player~$U$!) only picks a single block in each subround~$(r,i)$ of round~$r$.

\subparagraph*{Subrounds~$(r,0)$ and $(r,1)$}
Assume Player~$U$ picks $t_0^{\lag_0-1+r}$ in subround~$(r,0)$.
We fix~$\tr{t}_0^{\lag_0-1+r}$ such that $t_0^{\lag_0-1+r} \equivnew{1} \tr{t}_0^{\lag_0-1+r}$ and 
$
\size{\tr{t}_0^0\tr{t}_0^1 \cdots \tr{t}_0^{\lag_0-1+r}} \ge \size{\tr{t}_0^0\tr{t}_0^1 \cdots \tr{t}_0^{\lag_0+r-2}} + d
. 
$
This is possible, as the $\equivnew{1}$ equivalence class of $t_0^{\lag_0-1+r}$ is infinite and therefore contains arbitrarily long words.

We run $\tm_1$ on $\tr{t}_0^0 \tr{t}_0^1 \cdots \tr{t}_0^{\lag_0-1+r}$ and obtain another block~$\tr{t}_1^{\lag_0+r-2}$. 
There is a $t_1^{\lag_0+r-2}$ with $\size{t_1^{\lag_0+r-2}} = \size{t_0^{\lag_0+r-2}}$ such that $
\vectoroflengthtwo{t_0^{\lag_0+r-2}}{t_1^{\lag_0+r-2}} \equivnew{2} \vectoroflengthtwo{\tr{t}_0^{\lag_0+r-2}}{\tr{t}_1^{\lag_0+r-2}}$, as we have $t_0^{\lag_0+r-2} \equivnew{1} \tr{t}_0^{\lag_0+r-2}$.
We define $\sigma_1$ such that it picks the block~$t_1^{\lag_0+r-2}$ in subround~$(r,1)$ (note that $\lag_0+r-2 = \lag_1-1+r$).

\subparagraph*{Subrounds~$(r,i-1)$ and $(r,i)$ for $i \in \set{3,5, \ldots, k-1}$.}
Now, assume Player~$U$ picks $t_{i-1}^{\lag_{i-1}-1+r}$ in subround~$(r,i-1)$.
We fix $\tr{t}_{i-1}^{\lag_{i-1}-1+r}$ 
such that
$
\vectordotswithtwoargsatstart{\tr{t}_0^{\lag_{i-1}-1+r}}
 {\tr{t}_1^{\lag_{i-1}-1+r}}
 {\tr{t}_{i-1}^{\lag_{i-1}-1+r}}
  \equivnew{i}
 \vectordotswithtwoargsatstart{t_{0}^{\lag_{i-1}-1+r}}
 {t_{1}^{\lag_{i-1}-1+r}}
 {t_{i-1}^{\lag_{i-1}-1+r}}
 $. This is possible, as $
\vectordotswithtwoargsatstart{\tr{t}_0^{\lag_{i-1}-1+r}}
 {\tr{t}_1^{\lag_{i-1}-1+r}}
 {\tr{t}_{i-2}^{\lag_{i-1}-1+r}}
  \equivnew{i-1} 
 \vectordotswithtwoargsatstart{t_{0}^{\lag_{i-1}-1+r}}
 {t_{1}^{\lag_{i-1}-1+r}}
 {t_{i-2}^{\lag_{i-1}-1+r}}$ is an invariant of our construction.

 We run $\tm_i$ on 
\[
     \combine{\tr{t}_0^0\tr{t}_0^1\cdots\tr{t}_0^{\lag_{i-1}-1+r},
     \tr{t}_2^0\tr{t}_0^1\cdots\tr{t}_2^{\lag_{i-1}-1+r},
      \ldots , 
    \tr{t}_{i-1}^0\tr{t}_{i-1}^1\cdots\tr{t}_{i-1}^{\lag_{i-1}-1+r}},
\]
 yielding another block~$\tr{t}_{i}^{\lag_{i-1}+r-2}$. 
 As \[\vectordotswithtwoargsatstart
{\tr{t}_{0}^{\lag_{i-1}+r-2}}
{\tr{t}_{1}^{\lag_{i-1}+r-2}}
{\tr{t}_{i-1}^{\lag_{i-1}+r-2}}
  \equivnew{i}
\vectordotswithtwoargsatstart
{{t_{0}^{\lag_{i-1}+r-2}}}
{{t_{1}^{\lag_{i-1}+r-2}}}
{{t_{i-1}^{\lag_{i-1}+r-2}}}
 ,\] there also exists a 
 $t_{i}^{\lag_{i-1}+r-2}$ with $\size{t_{i}^{\lag_{i-1}+r-2}} = \size{t_{0}^{\lag_{i-1}+r-2}}$
such that $
\vectordotswithtwoargsatstart
{{t_{0}^{\lag_{i-1}+r-2}}}
{{t_{1}^{\lag_{i-1}+r-2}}}
{{t_{i}^{\lag_{i-1}+r-2}}}
\equivnew{i+1}
\vectordotswithtwoargsatstart
{\tr{t}_{0}^{\lag_{i-1}+r-2}}
{\tr{t}_{1}^{\lag_{i-1}+r-2}}
{\tr{t}_{i}^{\lag_{i-1}+r-2}}
 $.
We define $\sigma_i$ such that it picks $t_{i}^{\lag_{i-1}+r-2}$ in subround~$(r,i)$ (note that $\lag_{i-1}+r-2 = \lag_i-1+r$).

This completes the definition of the~$\sigma_i$. Note that each $\sigma_i$ does indeed only depend on the blocks picked in subrounds~$(r,j)$ with $j \in \set{0,2,\ldots, i-1}$, i.e., $\sigma_i$ is indeed a strategy for Player~$i$ in $\game(\tsys, \varphi)$.

It remains to show that the $\sigma_i$ are a winning collection of strategies. 
To this end, let $(t_0, t_1, \ldots, t_{k-1})$ be an outcome of a play that is consistent with the $\sigma_i$. 
If a $t_i$ with even $i$ is not in $\traces(\tsys)$, then the variable players win immediately. 
So, assume each $t_i$ with even~$i$ is in $\traces(\tsys)$. 
Let $\tr{t}_0, \tr{t}_1,\ldots, \tr{t}_{k-1}$ be the traces constructed during the inductive definition of the $\sigma_i$.
By applying Remark~\ref{lemma_equivkprops}.\ref{lemma_equivkprops_tsys}, we obtain that each $\tr{t}_i$ with even $i$ is in $\traces(\tsys)$ as well.
Also, the $\tr{t}_i$ for odd $i$ satisfy $\tr{t}_i = f_i(\tr{t}_0, \tr{t}_2, \ldots, \tr{t}_{i-1})$ by construction, i.e., they are obtained by applying the Skolem functions.
Hence, the variable assignment mapping $\pi_i$ to $\tr{t}_i$ satisfies $\psi$, which implies that $\combine{\tr{t}_0, \tr{t}_1,\ldots, \tr{t}_{i-1}}$ is in $L(\aut_\psi^\tsys)$.

Now, as $\vectordotswithtwoargsatstart{\tr{t}_0^n}{\tr{t}_1^n}{\tr{t}_{k-1}^n} \equivnew{k} \vectordotswithtwoargsatstart{{t_0^n}}{{t_1^n}}{{t_{k-1}^n}}$ for all $n$, applying Lemma~\ref{lemma_equivkprops}.\ref{lemma_equivkprops_autk} yields that $\combine{{t}_0, {t}_1,\ldots, {t}_{i-1}}$ is in $L(\aut_\psi^\tsys)$ as well, i.e., the variable players do indeed win.
\end{proof}

% \begin{proof}[Proof Sketch]
% \ref{lemma_correctness_fromfsws2compskolem}.) Let $(\sigma_i)_{i \in \set{1,3,\ldots, k-1}}$ be a winning collection of finite-state strategies. To construct a Turing machine computing a Skolem function for $\pi_i$, we use more and more blocks of the input traces given to the Turing machine to simulate a play of $\game(\tsys, \varphi)$, which yields an infinite sequence of blocks picked by Player~$i$. These are given as output by the Turing machine. Note that this simulation requires the Turing machine to simulate all $\sigma_j$ with $j \le i$, which can be done as the strategies are finite-state. These functions are indeed Skolem functions, as the strategies are winning. 

% \ref{lemma_correctness_fromcompskolem2ws}.) Let $f_i$ be a computable Skolem function for each $i \in \set{1,3,\ldots, k-1}$. To construct a strategy~$\sigma_i$ for Player~$i$ we would like to feed the blocks picked by Player~$U$ in subrounds~$(r,j)$ with even $j <i$ to the Turing machine implementing $f_i$. However, these might be too short to yield a long enough output. However, by construction, we can always pick words that are $\equivnew{i}^\refmark$-equivalent to these blocks and long enough to yield enough output. The collection of these strategies is winning, as the $f_i$ are Skolem functions. 
% \end{proof}

Now, one can formalize $\game(\tsys, \varphi)$ as a multi-player graph game of (hierarchical) imperfect information. The existence of a winning collection of strategies is decidable for such games~\cite{bbb,DBLP:conf/focs/PnueliR90}. Furthermore, if there is a winning collection of strategies, then also a winning collection of finite-state strategies, which closes the gap in the statement of Lemma~\ref{lemma_correctness}: $\tsys \models \varphi$ has computable Skolem functions iff the coalition of variable players has a winning collection of finite-state strategies.
Furthermore, similarly to the proof of Lemma~\ref{lemma_correctness}.\ref{lemma_correctness_fromfsws2compskolem}, one can show that such finite-state strategies can even be implemented by \bddfts, thereby completing the proof of Theorem~\ref{thm:main}.

\subsection{The Concrete Game}
\label{subsec_mainproof}

After having shown that the abstract game~$\game(\tsys, \varphi)$ characterizes the existence of computable Skolem functions, we now model $\game(\tsys, \varphi)$ as a multi-player graph game of imperfect information using the notation of Berwanger et al.~\cite{bbb}.
In Subsection~\ref{subsec_distributedgames}, we introduce the necessary definitions before we model the game in Subsection~\ref{subsec_concretegame}.
The games considered by Berwanger et al.\ are concurrent games (i.e., the players make their moves simultaneously), while $\game(\tsys, \varphi)$ is turn-based, i.e., the players make their moves one after the other. 
Hence, we will also introduce some notation for the special case of turn-based games, which simplifies our modeling.
%%%%%%%%%%%%%%%%%%%%%%%%%%%%%%%%%%%%%%%%%%%%%%%%%%%%%%%%%%%%%%%%%%%%%%%%%%%%%%%%%%%%%%%%%%%%%%%%%%%%%%%%%%%%%%%%%%%%%%%%
%%%%%%%%%%%%%%%%%%%%%%%%%%%%%%%%%%%%%%%%%%%%%%%%%%%%%%%%%%%%%%%%%%%%%%%%%%%%%%%%%%%%%%%%%%%%%%%%%%%%%%%%%%%%%%%%%%%%%%%%
%%%%%%%%%%%%%%%%%%%%%%%%%%%%%%%%%%%%%%%%%%%%%%%%%%%%%%%%%%%%%%%%%%%%%%%%%%%%%%%%%%%%%%%%%%%%%%%%%%%%%%%%%%%%%%%%%%%%%%%%
\subsubsection{Distributed Games}
\label{subsec_distributedgames}

\paragraph*{Game Graphs} 
Fix some set~$N = \set{1,\ldots, n}$ of players and a distinguished agent called Nature (which is not in $N$!). 
A profile is a list~$p = (p_1, \ldots, p_n)$ of elements~$p_i \in P_i$ for sets~$P_i$ that will be clear from context.
For each player~$i \in N$ we fix a finite set~$A^i$ of actions and a finite set~$B^i$ of observations.
A game graph~$G = (V, E, v_\initmark, (\beta^i)_{i\in N})$ consists of a finite set~$V$ of positions, an edge relation~$E \subseteq V \times A \times V$ representing simultaneous moves labeled by action profiles (i.e., $A = A^1 \times \cdots \times A^n$), an initial position~$v_\initmark \in V$, and a profile~$(\beta^i)_{i\in N}$ of observation functions~$\beta^i \colon V \rightarrow B^i$ that label, for  each player, the positions with observations.
We require that $E$ has no dead ends, i.e., for every $v \in V$ and every~$a \in A$ there is a $v' \in V$ with $(v,a,v') \in E$.

A game graph~$(V, E, v_\initmark, (\beta^i)_{i\in N})$ yields hierarchical information if there exists a total order~$\preceq$ over $N$ such that if $i \preceq j$ then for all $v,v' \in V$, $\beta^i(v) = \beta^i(v')$ implies $\beta^j(v) = \beta^j(v')$, i.e., if Player~$i$ cannot distinguish $v$ and $v'$, then neither can Player~$j$ for $i \preceq j$.

\paragraph*{Plays} 
Intuitively, a play starts at position~$v_\initmark \in V$ and proceeds in rounds. In a round at position~$v$, each Player~$i$ chooses simultaneously and independently an action~$a^i \in A^i$, then Nature chooses a successor position~$v'$ such that $(v,a,v') \in E$. Now, each player receives the observation~$\beta^i(v')$ and the next round is played at position~$v'$. 
Thus, a play of $G$ is an infinite sequence~$v_0 v_1 v_2 \cdots$ of vertices such that $v_0 = v_\initmark$ and for all $r \ge 0$ there is an $a_r \in A$ such that $(v_r, a_r, v_{r+1}) \in E$.

A history is a prefix~$v_0 v_1 \cdots v_r$ of a play. 
We denote the set of all histories by $\hist(G)$ and extend $\beta^i\colon V \rightarrow B^i$ to plays and histories by defining $\beta^i(v_0 v_1 v_2 \cdots) = \beta^i(v_1)\beta^i(v_2)\beta^i(v_3)\cdots$. Note that the observation of the initial position is discarded for technical reasons~\cite{bbb}.
We say two histories $h$ and $h'$ are indistinguishable to Player~$i$, denoted by $h \sim_i h'$, if $\beta^i(h) = \beta^i(h')$. 

\paragraph*{Strategies} 
A strategy for Player~$i$ is a mapping~$s^i \colon V^* \rightarrow A^i$ that satisfies $s^i(h) = s^i(h')$ for all $h,h'$ with $h \sim_i h'$ (i.e., the action selected by the strategy only depends on the observations of the history).
A play~$v_0 v_1 v_2 \cdots$ is consistent with $s^i$ if for every $r \ge 0$, there is an $a_r = (a^1, \ldots, a^n) \in A$ with $(v_r, a_r, v_{r+1}) \in E$ and $a^i = s^i(v_0 v_1 \cdots v_r)$.
A play is consistent with a strategy profile~$(s^1, \ldots,s^n )$ if it is consistent with each $s^i$.
The set of possible outcomes of a strategy profile is the set of all plays that are consistent with $s$.

A distributed game $\game = (G, W)$ consists of a game graph and a winning condition~$W \subseteq V^\omega$, where $V$ is the set of positions of $G$.
A play is winning if it is in $W$ and a strategy profile~$S$ is winning in $\game$ if all its outcomes are winning.

\paragraph*{Finite-state Strategies} 

Next, we define what it means for a strategy for Player~$i$ to be finite-state. So far, we used transducers, i.e., automata with output on transitions to implement strategies in a finitary manner.
From now on, we follow the definitions used by Berwanger et al.~\cite{bbb} and use Moore machines, i.e., finite automata with output on states. However, each Moore machine can be transformed into a transducer by \myquot{moving} the output from a state to all its outgoing transitions.

Let $A^i$ and $B^i$ be the actions and observations of Player~$i$. 
A Moore machine~$\moore = (M, m_\initmark, \upd,\nextmove)$ for Player~$i$ consists of a finite set~$M$ of memory states containing the initial memory state~$m_\initmark$, a memory update function~$\upd\colon M \times B^i \rightarrow M$, and a next-move function~$\nextmove\colon M\rightarrow A^i$.
We extend~$\upd$ to words over $B^i$ by defining $\upd(\epsilon) = m_\initmark$ and $\upd(b_0 b_1 \cdots b_r) = \upd(\upd(b_0b_1 \cdots b_{r-1}),b_r)$.
We say $\moore$ implements the strategy mapping $v_0 v_1 \cdots v_r$ to $\nextmove(\upd(\beta^i(v_0 v_1 \cdots v_r)))$. 
A strategy is finite-state, if it is implemented by some Moore machine.

\begin{proposition}[\cite{DBLP:conf/focs/PnueliR90,bbb}]
\label{prop_distributedgames}\hfill
\begin{enumerate}
    \item \label{prop_distributedgames_decidability}
        The following problem is decidable: Given a distributed game with $\omega$-regular winning condition, does it have a winning strategy profile?
    \item \label{prop_distributedgames_finitestate}
        A distributed game with $\omega$-regular winning condition has a winning strategy profile if and only if it has a winning profile of finite-state strategies.
\end{enumerate}
\end{proposition}

\paragraph*{Turn-based Game Graphs}
We say that a game graph~$G = (V, E, v_\initmark, (\beta^i)_{i\in N})$ is turn-based, if 
there is a function~$\owner \colon V \rightarrow N$ such that if $(v,a,v') \in E$, $(v,a',v'') \in E$, and the action profiles~$a$ and $a'$ having the same action for Player~$\owner(v)$ implies $v' = v''$. 
To simplify our notation, we label the edges leaving $v$ only by actions of Player~$\owner(v)$.
Thus, in a turn-based game graph, at every position~$v$ Player~$\owner(v)$ determines the possible next moves, and Nature selects one of them. Turn-based distributed games are distributed games whose game graphs are turn-based.

%%%%%%%%%%%%%%%%%%%%%%%%%%%%%%%%%%%%%%%%%%%%%%%%%%%%%%%%%%%%%%%%%%%%%%%%%%%%%%%%%%%%%%%%%%%%%%%%%%%%%%%%%%%%%%%%%%%%%%%%
%%%%%%%%%%%%%%%%%%%%%%%%%%%%%%%%%%%%%%%%%%%%%%%%%%%%%%%%%%%%%%%%%%%%%%%%%%%%%%%%%%%%%%%%%%%%%%%%%%%%%%%%%%%%%%%%%%%%%%%%
%%%%%%%%%%%%%%%%%%%%%%%%%%%%%%%%%%%%%%%%%%%%%%%%%%%%%%%%%%%%%%%%%%%%%%%%%%%%%%%%%%%%%%%%%%%%%%%%%%%%%%%%%%%%%%%%%%%%%%%%
\subsubsection{Formalization of the Abstract Game}
\label{subsec_concretegame}

Now, we are finally ready to formalize the abstract game~$\game(\tsys,\phi)$ described in Section~\ref{sec_abstractgame} as a turn-based distributed game.

We begin by introducing notation for the positions of the game, which intuitively keep track of blocks picked by the players of $\game(\tsys,\phi)$ until they can be processed by an automaton recognizing the winning condition. 
Due to the delay between the choices by the different players, this requires some notation.

Recall that we have defined $\lag_i$ for even $i$ to be the number of blocks Player~$U$ picks in subround~$(0,i)$ for variable~$i$ and for odd $i$ to be the number of blocks Player~$i$ picks in subround~$(0,i)$ for variable~$i$.
Now, let 
\[
D = \set{(j,x) \mid j \in \set{0,1,\ldots,k-1} \text{ and } x \in \set{0,1,\ldots,\lag_i-1} }.
\]
A configuration is a partial function~$c\colon D\rightarrow \blocks$, where $\blocks 
 = (\pow{\ap})^\ell$ denotes the set of blocks.
Let $C$ denote the set of all configurations.
The following definitions are visualized in Figure~\ref{fig_playevolutionconcrete}.

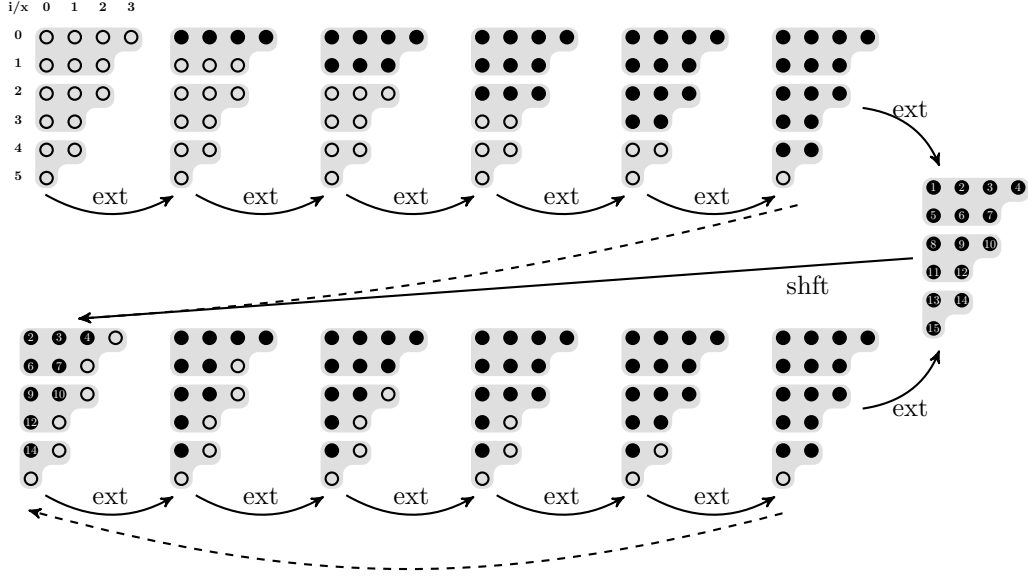
\begin{figure}[t]
    \centering

\begin{tikzpicture}[xscale=1.85,yscale=2]
    \node[anchor=south] (i0) at (0,0) {\scalebox{.5}{\begin{tikzpicture}[ultra thick]

\def\x{.75cm}
\def\y{-.75cm}
%% Player 1
% round 0
\fill[rounded corners=7, gray!25] (.6*\x, -.6*\y) -- (4.4*\x, -.6*\y) -- (4.4*\x, .25*\y) -- (3.4*\x,.25*\y)  -- (3.4*\x,1.1*\y) -- (.6*\x,1.1*\y) -- cycle;

%% Player 3
% round 0
\fill[rounded corners=7, gray!25] (.6*\x, 1.4*\y) -- (3.4*\x, 1.4*\y) -- (3.4*\x, 2.25*\y) -- (2.4*\x,2.25*\y)  -- (2.4*\x,3.1*\y) -- (.6*\x,3.1*\y) -- cycle;

% Player 5
% round 0
\fill[rounded corners=7, gray!25] (.6*\x, 3.4*\y) -- (2.4*\x, 3.4*\y) -- (2.4*\x, 4.25*\y) -- (1.4*\x,4.25*\y)  -- (1.4*\x,5.1*\y) -- (.6*\x,5.1*\y) -- cycle;

   \node at (0, 0) {\bf 0};
    \node at (0, \y) {\bf 1};
    \node at (0, 2*\y) {\bf 2};
    \node at (0, 3*\y) {\bf 3};
    \node at (0, 4*\y) {\bf 4};
    \node at (0, 5*\y) {\bf 5};

    \node[anchor = center] at (0,-1.3*\y) {\bf i/x};

    \node at (\x,-\y) {\bf 0};
    \node at (2*\x,-\y) {\bf 1};
    \node at (3*\x,-\y) {\bf 2};
    \node at (4*\x,-\y) {\bf 3};
    
    % \node at (0, 0) {0};
    % \node at (0, \y) {1};
    % \node at (0, 2*\y) {2};
    % \node at (0, 3*\y) {3};
    % \node at (0, 4*\y) {4};
    % \node at (0, 5*\y) {5};

    % \node at (0,-\y) {i/r};

    % \node at (\x,-\y) {0};
    % \node at (2*\x,-\y) {1};
    % \node at (3*\x,-\y) {2};
    % \node at (4*\x,-\y) {3};

    \node[nind] at (\x,0) {};
    \node[nind] at (\x,\y) {};
    \node[nind] at (\x,2*\y) {};
    \node[nind] at (\x,3*\y) {};
    \node[nind] at (\x,4*\y) {};
    \node[nind] at (\x,5*\y) {};
    
    \node[nind] at (2*\x,0) {};
    \node[nind] at (2*\x,\y) {};
    \node[nind] at (2*\x,2*\y) {};
    \node[nind] at (2*\x,3*\y) {};
    \node[nind] at (2*\x,4*\y) {};
    
    \node[nind] at (3*\x,0) {};
    \node[nind] at (3*\x,\y) {};
    \node[nind] at (3*\x,2*\y) {};

    \node[nind] at (4*\x,0) {};

    \end{tikzpicture}}};
    \node[anchor=south] (i1) at (1,0) {\scalebox{.5}{\begin{tikzpicture}[ultra thick]

\def\x{.75cm}
\def\y{-.75cm}
%% Player 1
% round 0
\fill[rounded corners=7, gray!25] (.6*\x, -.6*\y) -- (4.4*\x, -.6*\y) -- (4.4*\x, .25*\y) -- (3.4*\x,.25*\y)  -- (3.4*\x,1.1*\y) -- (.6*\x,1.1*\y) -- cycle;

%% Player 3
% round 0
\fill[rounded corners=7, gray!25] (.6*\x, 1.4*\y) -- (3.4*\x, 1.4*\y) -- (3.4*\x, 2.25*\y) -- (2.4*\x,2.25*\y)  -- (2.4*\x,3.1*\y) -- (.6*\x,3.1*\y) -- cycle;

% Player 5
% round 0
\fill[rounded corners=7, gray!25] (.6*\x, 3.4*\y) -- (2.4*\x, 3.4*\y) -- (2.4*\x, 4.25*\y) -- (1.4*\x,4.25*\y)  -- (1.4*\x,5.1*\y) -- (.6*\x,5.1*\y) -- cycle;

    % \node at (0, 0) {0};
    % \node at (0, \y) {1};
    % \node at (0, 2*\y) {2};
    % \node at (0, 3*\y) {3};
    % \node at (0, 4*\y) {4};
    % \node at (0, 5*\y) {5};

    % \node at (0,-\y) {i/r};

    % \node at (\x,-\y) {0};
    % \node at (2*\x,-\y) {1};
    % \node at (3*\x,-\y) {2};
    % \node at (4*\x,-\y) {3};

    \node[ind] at (\x,0) {};
    \node[nind] at (\x,\y) {};
    \node[nind] at (\x,2*\y) {};
    \node[nind] at (\x,3*\y) {};
    \node[nind] at (\x,4*\y) {};
    \node[nind] at (\x,5*\y) {};
    
    \node[ind] at (2*\x,0) {};
    \node[nind] at (2*\x,\y) {};
    \node[nind] at (2*\x,2*\y) {};
    \node[nind] at (2*\x,3*\y) {};
    \node[nind] at (2*\x,4*\y) {};
    
    \node[ind] at (3*\x,0) {};
    \node[nind] at (3*\x,\y) {};
    \node[nind] at (3*\x,2*\y) {};

    \node[ind] at (4*\x,0) {};

    \end{tikzpicture}}};
    \node[anchor=south] (i2) at (2,0) {\scalebox{.5}{\begin{tikzpicture}[ultra thick]

\def\x{.75cm}
\def\y{-.75cm}
%% Player 1
% round 0
\fill[rounded corners=7, gray!25] (.6*\x, -.6*\y) -- (4.4*\x, -.6*\y) -- (4.4*\x, .25*\y) -- (3.4*\x,.25*\y)  -- (3.4*\x,1.1*\y) -- (.6*\x,1.1*\y) -- cycle;

%% Player 3
% round 0
\fill[rounded corners=7, gray!25] (.6*\x, 1.4*\y) -- (3.4*\x, 1.4*\y) -- (3.4*\x, 2.25*\y) -- (2.4*\x,2.25*\y)  -- (2.4*\x,3.1*\y) -- (.6*\x,3.1*\y) -- cycle;

% Player 5
% round 0
\fill[rounded corners=7, gray!25] (.6*\x, 3.4*\y) -- (2.4*\x, 3.4*\y) -- (2.4*\x, 4.25*\y) -- (1.4*\x,4.25*\y)  -- (1.4*\x,5.1*\y) -- (.6*\x,5.1*\y) -- cycle;

    % \node at (0, 0) {0};
    % \node at (0, \y) {1};
    % \node at (0, 2*\y) {2};
    % \node at (0, 3*\y) {3};
    % \node at (0, 4*\y) {4};
    % \node at (0, 5*\y) {5};

    % \node at (0,-\y) {i/r};

    % \node at (\x,-\y) {0};
    % \node at (2*\x,-\y) {1};
    % \node at (3*\x,-\y) {2};
    % \node at (4*\x,-\y) {3};

    \node[ind] at (\x,0) {};
    \node[ind] at (\x,\y) {};
    \node[nind] at (\x,2*\y) {};
    \node[nind] at (\x,3*\y) {};
    \node[nind] at (\x,4*\y) {};
    \node[nind] at (\x,5*\y) {};
    
    \node[ind] at (2*\x,0) {};
    \node[ind] at (2*\x,\y) {};
    \node[nind] at (2*\x,2*\y) {};
    \node[nind] at (2*\x,3*\y) {};
    \node[nind] at (2*\x,4*\y) {};
    
    \node[ind] at (3*\x,0) {};
    \node[ind] at (3*\x,\y) {};
    \node[nind] at (3*\x,2*\y) {};

    \node[ind] at (4*\x,0) {};

    \end{tikzpicture}}};
    \node[anchor=south] (i3) at (3,0) {\scalebox{.5}{\begin{tikzpicture}[ultra thick]

\def\x{.75cm}
\def\y{-.75cm}
%% Player 1
% round 0
\fill[rounded corners=7, gray!25] (.6*\x, -.6*\y) -- (4.4*\x, -.6*\y) -- (4.4*\x, .25*\y) -- (3.4*\x,.25*\y)  -- (3.4*\x,1.1*\y) -- (.6*\x,1.1*\y) -- cycle;

%% Player 3
% round 0
\fill[rounded corners=7, gray!25] (.6*\x, 1.4*\y) -- (3.4*\x, 1.4*\y) -- (3.4*\x, 2.25*\y) -- (2.4*\x,2.25*\y)  -- (2.4*\x,3.1*\y) -- (.6*\x,3.1*\y) -- cycle;

% Player 5
% round 0
\fill[rounded corners=7, gray!25] (.6*\x, 3.4*\y) -- (2.4*\x, 3.4*\y) -- (2.4*\x, 4.25*\y) -- (1.4*\x,4.25*\y)  -- (1.4*\x,5.1*\y) -- (.6*\x,5.1*\y) -- cycle;

    % \node at (0, 0) {0};
    % \node at (0, \y) {1};
    % \node at (0, 2*\y) {2};
    % \node at (0, 3*\y) {3};
    % \node at (0, 4*\y) {4};
    % \node at (0, 5*\y) {5};

    % \node at (0,-\y) {i/r};

    % \node at (\x,-\y) {0};
    % \node at (2*\x,-\y) {1};
    % \node at (3*\x,-\y) {2};
    % \node at (4*\x,-\y) {3};

    \node[ind] at (\x,0) {};
    \node[ind] at (\x,\y) {};
    \node[ind] at (\x,2*\y) {};
    \node[nind] at (\x,3*\y) {};
    \node[nind] at (\x,4*\y) {};
    \node[nind] at (\x,5*\y) {};
    
    \node[ind] at (2*\x,0) {};
    \node[ind] at (2*\x,\y) {};
    \node[ind] at (2*\x,2*\y) {};
    \node[nind] at (2*\x,3*\y) {};
    \node[nind] at (2*\x,4*\y) {};
    
    \node[ind] at (3*\x,0) {};
    \node[ind] at (3*\x,\y) {};
    \node[ind] at (3*\x,2*\y) {};

    \node[ind] at (4*\x,0) {};

    \end{tikzpicture}}};
    \node[anchor=south] (i4) at (4,0) {\scalebox{.5}{\begin{tikzpicture}[ultra thick]

\def\x{.75cm}
\def\y{-.75cm}
%% Player 1
% round 0
\fill[rounded corners=7, gray!25] (.6*\x, -.6*\y) -- (4.4*\x, -.6*\y) -- (4.4*\x, .25*\y) -- (3.4*\x,.25*\y)  -- (3.4*\x,1.1*\y) -- (.6*\x,1.1*\y) -- cycle;

%% Player 3
% round 0
\fill[rounded corners=7, gray!25] (.6*\x, 1.4*\y) -- (3.4*\x, 1.4*\y) -- (3.4*\x, 2.25*\y) -- (2.4*\x,2.25*\y)  -- (2.4*\x,3.1*\y) -- (.6*\x,3.1*\y) -- cycle;

% Player 5
% round 0
\fill[rounded corners=7, gray!25] (.6*\x, 3.4*\y) -- (2.4*\x, 3.4*\y) -- (2.4*\x, 4.25*\y) -- (1.4*\x,4.25*\y)  -- (1.4*\x,5.1*\y) -- (.6*\x,5.1*\y) -- cycle;

    % \node at (0, 0) {0};
    % \node at (0, \y) {1};
    % \node at (0, 2*\y) {2};
    % \node at (0, 3*\y) {3};
    % \node at (0, 4*\y) {4};
    % \node at (0, 5*\y) {5};

    % \node at (0,-\y) {i/r};

    % \node at (\x,-\y) {0};
    % \node at (2*\x,-\y) {1};
    % \node at (3*\x,-\y) {2};
    % \node at (4*\x,-\y) {3};

    \node[ind] at (\x,0) {};
    \node[ind] at (\x,\y) {};
    \node[ind] at (\x,2*\y) {};
    \node[ind] at (\x,3*\y) {};
    \node[nind] at (\x,4*\y) {};
    \node[nind] at (\x,5*\y) {};
    
    \node[ind] at (2*\x,0) {};
    \node[ind] at (2*\x,\y) {};
    \node[ind] at (2*\x,2*\y) {};
    \node[ind] at (2*\x,3*\y) {};
    \node[nind] at (2*\x,4*\y) {};
    
    \node[ind] at (3*\x,0) {};
    \node[ind] at (3*\x,\y) {};
    \node[ind] at (3*\x,2*\y) {};

    \node[ind] at (4*\x,0) {};

    \end{tikzpicture}}};
    \node[anchor=south] (i5) at (5,0) {\scalebox{.5}{\begin{tikzpicture}[ultra thick]

\def\x{.75cm}
\def\y{-.75cm}
%% Player 1
% round 0
\fill[rounded corners=7, gray!25] (.6*\x, -.6*\y) -- (4.4*\x, -.6*\y) -- (4.4*\x, .25*\y) -- (3.4*\x,.25*\y)  -- (3.4*\x,1.1*\y) -- (.6*\x,1.1*\y) -- cycle;

%% Player 3
% round 0
\fill[rounded corners=7, gray!25] (.6*\x, 1.4*\y) -- (3.4*\x, 1.4*\y) -- (3.4*\x, 2.25*\y) -- (2.4*\x,2.25*\y)  -- (2.4*\x,3.1*\y) -- (.6*\x,3.1*\y) -- cycle;

% Player 5
% round 0
\fill[rounded corners=7, gray!25] (.6*\x, 3.4*\y) -- (2.4*\x, 3.4*\y) -- (2.4*\x, 4.25*\y) -- (1.4*\x,4.25*\y)  -- (1.4*\x,5.1*\y) -- (.6*\x,5.1*\y) -- cycle;

    % \node at (0, 0) {0};
    % \node at (0, \y) {1};
    % \node at (0, 2*\y) {2};
    % \node at (0, 3*\y) {3};
    % \node at (0, 4*\y) {4};
    % \node at (0, 5*\y) {5};

    % \node at (0,-\y) {i/r};

    % \node at (\x,-\y) {0};
    % \node at (2*\x,-\y) {1};
    % \node at (3*\x,-\y) {2};
    % \node at (4*\x,-\y) {3};

    \node[ind] at (\x,0) {};
    \node[ind] at (\x,\y) {};
    \node[ind] at (\x,2*\y) {};
    \node[ind] at (\x,3*\y) {};
    \node[ind] at (\x,4*\y) {};
    \node[nind] at (\x,5*\y) {};
    
    \node[ind] at (2*\x,0) {};
    \node[ind] at (2*\x,\y) {};
    \node[ind] at (2*\x,2*\y) {};
    \node[ind] at (2*\x,3*\y) {};
    \node[ind] at (2*\x,4*\y) {};
    
    \node[ind] at (3*\x,0) {};
    \node[ind] at (3*\x,\y) {};
    \node[ind] at (3*\x,2*\y) {};

    \node[ind] at (4*\x,0) {};

    \end{tikzpicture}}};
    \node[anchor=south] (f) at (6,-1) {\scalebox{.5}{\begin{tikzpicture}[ultra thick]

\def\x{.75cm}
\def\y{-.75cm}
%% Player 1
% round 0
\fill[rounded corners=7, gray!25] (.6*\x, -.6*\y) -- (4.4*\x, -.6*\y) -- (4.4*\x, .25*\y) -- (3.4*\x,.25*\y)  -- (3.4*\x,1.1*\y) -- (.6*\x,1.1*\y) -- cycle;

%% Player 3
% round 0
\fill[rounded corners=7, gray!25] (.6*\x, 1.4*\y) -- (3.4*\x, 1.4*\y) -- (3.4*\x, 2.25*\y) -- (2.4*\x,2.25*\y)  -- (2.4*\x,3.1*\y) -- (.6*\x,3.1*\y) -- cycle;

% Player 5
% round 0
\fill[rounded corners=7, gray!25] (.6*\x, 3.4*\y) -- (2.4*\x, 3.4*\y) -- (2.4*\x, 4.25*\y) -- (1.4*\x,4.25*\y)  -- (1.4*\x,5.1*\y) -- (.6*\x,5.1*\y) -- cycle;

    \node[ind] at (\x,0) {};
    \node[ind] at (\x,\y) {};
    \node[ind] at (\x,2*\y) {};
    \node[ind] at (\x,3*\y) {};
    \node[ind] at (\x,4*\y) {};
    \node[ind] at (\x,5*\y) {};
    
    \node[ind] at (2*\x,0) {};
    \node[ind] at (2*\x,\y) {};
    \node[ind] at (2*\x,2*\y) {};
    \node[ind] at (2*\x,3*\y) {};
    \node[ind] at (2*\x,4*\y) {};
    
    \node[ind] at (3*\x,0) {};
    \node[ind] at (3*\x,\y) {};
    \node[ind] at (3*\x,2*\y) {};

    \node[ind] at (4*\x,0) {};

    \node[wt] at (\x,0) {1};
    \node[wt] at (\x,\y) {5};
    \node[wt] at (\x,2*\y) {8};
    \node[wt] at (\x,3*\y) {11};
    \node[wt] at (\x,4*\y) {13};
    \node[wt] at (\x,5*\y) {15};
    
    \node[wt] at (2*\x,0) {2};
    \node[wt] at (2*\x,\y) {6};
    \node[wt] at (2*\x,2*\y) {9};
    \node[wt] at (2*\x,3*\y) {12};
    \node[wt] at (2*\x,4*\y) {14};
    
    \node[wt] at (3*\x,0) {3};
    \node[wt] at (3*\x,\y) {7};
    \node[wt] at (3*\x,2*\y) {10};

    \node[wt] at (4*\x,0) {4};

    \end{tikzpicture}}};
    \node[anchor=south] (s0) at (0,-2) {\scalebox{.5}{\begin{tikzpicture}[ultra thick]

\def\x{.75cm}
\def\y{-.75cm}
%% Player 1
% round 0
\fill[rounded corners=7, gray!25] (.6*\x, -.6*\y) -- (4.4*\x, -.6*\y) -- (4.4*\x, .25*\y) -- (3.4*\x,.25*\y)  -- (3.4*\x,1.1*\y) -- (.6*\x,1.1*\y) -- cycle;

%% Player 3
% round 0
\fill[rounded corners=7, gray!25] (.6*\x, 1.4*\y) -- (3.4*\x, 1.4*\y) -- (3.4*\x, 2.25*\y) -- (2.4*\x,2.25*\y)  -- (2.4*\x,3.1*\y) -- (.6*\x,3.1*\y) -- cycle;

% Player 5
% round 0
\fill[rounded corners=7, gray!25] (.6*\x, 3.4*\y) -- (2.4*\x, 3.4*\y) -- (2.4*\x, 4.25*\y) -- (1.4*\x,4.25*\y)  -- (1.4*\x,5.1*\y) -- (.6*\x,5.1*\y) -- cycle;

    % \node at (0, 0) {0};
    % \node at (0, \y) {1};
    % \node at (0, 2*\y) {2};
    % \node at (0, 3*\y) {3};
    % \node at (0, 4*\y) {4};
    % \node at (0, 5*\y) {5};

    % \node at (0,-\y) {i/r};

    % \node at (\x,-\y) {0};
    % \node at (2*\x,-\y) {1};
    % \node at (3*\x,-\y) {2};
    % \node at (4*\x,-\y) {3};

    \node[ind] at (\x,0) {};
    \node[ind] at (\x,\y) {};
    \node[ind] at (\x,2*\y) {};
    \node[ind] at (\x,3*\y) {};
    \node[ind] at (\x,4*\y) {};
    \node[nind] at (\x,5*\y) {};
    
    \node[ind] at (2*\x,0) {};
    \node[ind] at (2*\x,\y) {};
    \node[ind] at (2*\x,2*\y) {};
    \node[nind] at (2*\x,3*\y) {};
    \node[nind] at (2*\x,4*\y) {};
    
    \node[ind] at (3*\x,0) {};
    \node[nind] at (3*\x,\y) {};
    \node[nind] at (3*\x,2*\y) {};

    \node[nind] at (4*\x,0) {};

    \node[wt] at (\x,0) {2};
    \node[wt] at (\x,\y) {6};
    \node[wt] at (\x,2*\y) {9};
    \node[wt] at (\x,3*\y) {12};
    \node[wt] at (\x,4*\y) {14};
    
    \node[wt] at (2*\x,0) {3};
    \node[wt] at (2*\x,\y) {7};
    \node[wt] at (2*\x,2*\y) {10};

    \node[wt] at (3*\x,0) {4};

    \end{tikzpicture}}};
    \node[anchor=south] (s1) at (1,-2) {\scalebox{.5}{\begin{tikzpicture}[ultra thick]

\def\x{.75cm}
\def\y{-.75cm}
%% Player 1
% round 0
\fill[rounded corners=7, gray!25] (.6*\x, -.6*\y) -- (4.4*\x, -.6*\y) -- (4.4*\x, .25*\y) -- (3.4*\x,.25*\y)  -- (3.4*\x,1.1*\y) -- (.6*\x,1.1*\y) -- cycle;

%% Player 3
% round 0
\fill[rounded corners=7, gray!25] (.6*\x, 1.4*\y) -- (3.4*\x, 1.4*\y) -- (3.4*\x, 2.25*\y) -- (2.4*\x,2.25*\y)  -- (2.4*\x,3.1*\y) -- (.6*\x,3.1*\y) -- cycle;

% Player 5
% round 0
\fill[rounded corners=7, gray!25] (.6*\x, 3.4*\y) -- (2.4*\x, 3.4*\y) -- (2.4*\x, 4.25*\y) -- (1.4*\x,4.25*\y)  -- (1.4*\x,5.1*\y) -- (.6*\x,5.1*\y) -- cycle;

    % \node at (0, 0) {0};
    % \node at (0, \y) {1};
    % \node at (0, 2*\y) {2};
    % \node at (0, 3*\y) {3};
    % \node at (0, 4*\y) {4};
    % \node at (0, 5*\y) {5};

    % \node at (0,-\y) {i/r};

    % \node at (\x,-\y) {0};
    % \node at (2*\x,-\y) {1};
    % \node at (3*\x,-\y) {2};
    % \node at (4*\x,-\y) {3};

    \node[ind] at (\x,0) {};
    \node[ind] at (\x,\y) {};
    \node[ind] at (\x,2*\y) {};
    \node[ind] at (\x,3*\y) {};
    \node[ind] at (\x,4*\y) {};
    \node[nind] at (\x,5*\y) {};
    
    \node[ind] at (2*\x,0) {};
    \node[ind] at (2*\x,\y) {};
    \node[ind] at (2*\x,2*\y) {};
    \node[nind] at (2*\x,3*\y) {};
    \node[nind] at (2*\x,4*\y) {};
    
    \node[ind] at (3*\x,0) {};
    \node[nind] at (3*\x,\y) {};
    \node[nind] at (3*\x,2*\y) {};

    \node[ind] at (4*\x,0) {};

    \end{tikzpicture}}};
    \node[anchor=south] (s2) at (2,-2) {\scalebox{.5}{\begin{tikzpicture}[ultra thick]

\def\x{.75cm}
\def\y{-.75cm}
%% Player 1
% round 0
\fill[rounded corners=7, gray!25] (.6*\x, -.6*\y) -- (4.4*\x, -.6*\y) -- (4.4*\x, .25*\y) -- (3.4*\x,.25*\y)  -- (3.4*\x,1.1*\y) -- (.6*\x,1.1*\y) -- cycle;

%% Player 3
% round 0
\fill[rounded corners=7, gray!25] (.6*\x, 1.4*\y) -- (3.4*\x, 1.4*\y) -- (3.4*\x, 2.25*\y) -- (2.4*\x,2.25*\y)  -- (2.4*\x,3.1*\y) -- (.6*\x,3.1*\y) -- cycle;

% Player 5
% round 0
\fill[rounded corners=7, gray!25] (.6*\x, 3.4*\y) -- (2.4*\x, 3.4*\y) -- (2.4*\x, 4.25*\y) -- (1.4*\x,4.25*\y)  -- (1.4*\x,5.1*\y) -- (.6*\x,5.1*\y) -- cycle;

    % \node at (0, 0) {0};
    % \node at (0, \y) {1};
    % \node at (0, 2*\y) {2};
    % \node at (0, 3*\y) {3};
    % \node at (0, 4*\y) {4};
    % \node at (0, 5*\y) {5};

    % \node at (0,-\y) {i/r};

    % \node at (\x,-\y) {0};
    % \node at (2*\x,-\y) {1};
    % \node at (3*\x,-\y) {2};
    % \node at (4*\x,-\y) {3};

    \node[ind] at (\x,0) {};
    \node[ind] at (\x,\y) {};
    \node[ind] at (\x,2*\y) {};
    \node[ind] at (\x,3*\y) {};
    \node[ind] at (\x,4*\y) {};
    \node[nind] at (\x,5*\y) {};
    
    \node[ind] at (2*\x,0) {};
    \node[ind] at (2*\x,\y) {};
    \node[ind] at (2*\x,2*\y) {};
    \node[nind] at (2*\x,3*\y) {};
    \node[nind] at (2*\x,4*\y) {};
    
    \node[ind] at (3*\x,0) {};
    \node[ind] at (3*\x,\y) {};
    \node[nind] at (3*\x,2*\y) {};

    \node[ind] at (4*\x,0) {};

    \end{tikzpicture}}};
    \node[anchor=south] (s3) at (3,-2) {\scalebox{.5}{\begin{tikzpicture}[ultra thick]

\def\x{.75cm}
\def\y{-.75cm}
%% Player 1
% round 0
\fill[rounded corners=7, gray!25] (.6*\x, -.6*\y) -- (4.4*\x, -.6*\y) -- (4.4*\x, .25*\y) -- (3.4*\x,.25*\y)  -- (3.4*\x,1.1*\y) -- (.6*\x,1.1*\y) -- cycle;

%% Player 3
% round 0
\fill[rounded corners=7, gray!25] (.6*\x, 1.4*\y) -- (3.4*\x, 1.4*\y) -- (3.4*\x, 2.25*\y) -- (2.4*\x,2.25*\y)  -- (2.4*\x,3.1*\y) -- (.6*\x,3.1*\y) -- cycle;

% Player 5
% round 0
\fill[rounded corners=7, gray!25] (.6*\x, 3.4*\y) -- (2.4*\x, 3.4*\y) -- (2.4*\x, 4.25*\y) -- (1.4*\x,4.25*\y)  -- (1.4*\x,5.1*\y) -- (.6*\x,5.1*\y) -- cycle;

    % \node at (0, 0) {0};
    % \node at (0, \y) {1};
    % \node at (0, 2*\y) {2};
    % \node at (0, 3*\y) {3};
    % \node at (0, 4*\y) {4};
    % \node at (0, 5*\y) {5};

    % \node at (0,-\y) {i/r};

    % \node at (\x,-\y) {0};
    % \node at (2*\x,-\y) {1};
    % \node at (3*\x,-\y) {2};
    % \node at (4*\x,-\y) {3};
    
        \node[ind] at (\x,0) {};
    \node[ind] at (\x,\y) {};
    \node[ind] at (\x,2*\y) {};
    \node[ind] at (\x,3*\y) {};
    \node[ind] at (\x,4*\y) {};
    \node[nind] at (\x,5*\y) {};
    
    \node[ind] at (2*\x,0) {};
    \node[ind] at (2*\x,\y) {};
    \node[ind] at (2*\x,2*\y) {};
    \node[nind] at (2*\x,3*\y) {};
    \node[nind] at (2*\x,4*\y) {};
    
    \node[ind] at (3*\x,0) {};
    \node[ind] at (3*\x,\y) {};
    \node[ind] at (3*\x,2*\y) {};

    \node[ind] at (4*\x,0) {};

    \end{tikzpicture}}};
    \node[anchor=south] (s4) at (4,-2) {\scalebox{.5}{\begin{tikzpicture}[ultra thick]

\def\x{.75cm}
\def\y{-.75cm}
%% Player 1
% round 0
\fill[rounded corners=7, gray!25] (.6*\x, -.6*\y) -- (4.4*\x, -.6*\y) -- (4.4*\x, .25*\y) -- (3.4*\x,.25*\y)  -- (3.4*\x,1.1*\y) -- (.6*\x,1.1*\y) -- cycle;

%% Player 3
% round 0
\fill[rounded corners=7, gray!25] (.6*\x, 1.4*\y) -- (3.4*\x, 1.4*\y) -- (3.4*\x, 2.25*\y) -- (2.4*\x,2.25*\y)  -- (2.4*\x,3.1*\y) -- (.6*\x,3.1*\y) -- cycle;

% Player 5
% round 0
\fill[rounded corners=7, gray!25] (.6*\x, 3.4*\y) -- (2.4*\x, 3.4*\y) -- (2.4*\x, 4.25*\y) -- (1.4*\x,4.25*\y)  -- (1.4*\x,5.1*\y) -- (.6*\x,5.1*\y) -- cycle;

    % \node at (0, 0) {0};
    % \node at (0, \y) {1};
    % \node at (0, 2*\y) {2};
    % \node at (0, 3*\y) {3};
    % \node at (0, 4*\y) {4};
    % \node at (0, 5*\y) {5};

    % \node at (0,-\y) {i/r};

    % \node at (\x,-\y) {0};
    % \node at (2*\x,-\y) {1};
    % \node at (3*\x,-\y) {2};
    % \node at (4*\x,-\y) {3};
    
        \node[ind] at (\x,0) {};
    \node[ind] at (\x,\y) {};
    \node[ind] at (\x,2*\y) {};
    \node[ind] at (\x,3*\y) {};
    \node[ind] at (\x,4*\y) {};
    \node[nind] at (\x,5*\y) {};
    
    \node[ind] at (2*\x,0) {};
    \node[ind] at (2*\x,\y) {};
    \node[ind] at (2*\x,2*\y) {};
    \node[ind] at (2*\x,3*\y) {};
    \node[nind] at (2*\x,4*\y) {};
    
    \node[ind] at (3*\x,0) {};
    \node[ind] at (3*\x,\y) {};
    \node[ind] at (3*\x,2*\y) {};

    \node[ind] at (4*\x,0) {};

    \end{tikzpicture}}};
    \node[anchor=south] (s5) at (5,-2) {\scalebox{.5}{\begin{tikzpicture}[ultra thick]

\def\x{.75cm}
\def\y{-.75cm}
%% Player 1
% round 0
\fill[rounded corners=7, gray!25] (.6*\x, -.6*\y) -- (4.4*\x, -.6*\y) -- (4.4*\x, .25*\y) -- (3.4*\x,.25*\y)  -- (3.4*\x,1.1*\y) -- (.6*\x,1.1*\y) -- cycle;

%% Player 3
% round 0
\fill[rounded corners=7, gray!25] (.6*\x, 1.4*\y) -- (3.4*\x, 1.4*\y) -- (3.4*\x, 2.25*\y) -- (2.4*\x,2.25*\y)  -- (2.4*\x,3.1*\y) -- (.6*\x,3.1*\y) -- cycle;

% Player 5
% round 0
\fill[rounded corners=7, gray!25] (.6*\x, 3.4*\y) -- (2.4*\x, 3.4*\y) -- (2.4*\x, 4.25*\y) -- (1.4*\x,4.25*\y)  -- (1.4*\x,5.1*\y) -- (.6*\x,5.1*\y) -- cycle;

    % \node at (0, 0) {0};
    % \node at (0, \y) {1};
    % \node at (0, 2*\y) {2};
    % \node at (0, 3*\y) {3};
    % \node at (0, 4*\y) {4};
    % \node at (0, 5*\y) {5};

    % \node at (0,-\y) {i/r};

    % \node at (\x,-\y) {0};
    % \node at (2*\x,-\y) {1};
    % \node at (3*\x,-\y) {2};
    % \node at (4*\x,-\y) {3};

        \node[ind] at (\x,0) {};
    \node[ind] at (\x,\y) {};
    \node[ind] at (\x,2*\y) {};
    \node[ind] at (\x,3*\y) {};
    \node[ind] at (\x,4*\y) {};
    \node[nind] at (\x,5*\y) {};
    
    \node[ind] at (2*\x,0) {};
    \node[ind] at (2*\x,\y) {};
    \node[ind] at (2*\x,2*\y) {};
    \node[ind] at (2*\x,3*\y) {};
    \node[ind] at (2*\x,4*\y) {};
    
    \node[ind] at (3*\x,0) {};
    \node[ind] at (3*\x,\y) {};
    \node[ind] at (3*\x,2*\y) {};

    \node[ind] at (4*\x,0) {};

    \end{tikzpicture}}};

    \path[->]
    ($(i0.south) + (-.2,.02)$) edge[bend right] node[above] {$\extend$} ($(i1.south west) + (.1,.02)$)
    ($(i1.south) + (-.2,.02)$) edge[bend right] node[above] {$\extend$} ($(i2.south west) + (.1,.02)$)
    ($(i2.south) + (-.2,.02)$) edge[bend right] node[above] {$\extend$} ($(i3.south west) + (.1,.02)$)
    ($(i3.south) + (-.2,.02)$) edge[bend right] node[above] {$\extend$} ($(i4.south west) + (.1,.02)$)
    ($(i4.south) + (-.2,.02)$) edge[bend right] node[above] {$\extend$} ($(i5.south west) + (.1,.02)$)
    ($(i5.east) + (-.2,0)$) edge[bend left] node[above] {$\extend$} ($(f.north) + (-.25,0)$)
    ($ (i5.south)+(-.2,-.05)$) edge[dashed,bend left=5] node[above] {} (s0.north)
     (f.west) edge node[below,very near start] {$\free$} (s0.north)
    ($(s0.south) + (-.2,.02)$) edge[bend right] node[above] {$\extend$} ($(s1.south west) + (.1,.02)$)
    ($(s1.south) + (-.2,.02)$) edge[bend right] node[above] {$\extend$} ($(s2.south west) + (.1,.02)$)
    ($(s2.south) + (-.2,.02)$) edge[bend right] node[above] {$\extend$} ($(s3.south west) + (.1,.02)$)
    ($(s3.south) + (-.2,.02)$) edge[bend right] node[above] {$\extend$} ($(s4.south west) + (.1,.02)$)
    ($(s4.south) + (-.2,.02)$) edge[bend right] node[above] {$\extend$} ($(s5.south west) + (.1,.02)$)
    ($(s5.east) + (-.2,0)$) edge[bend right] node[below] {$\extend$} ($(f.south) + (-.25,0)$)
    ($(s5.south) + (-.3,-.1) $) edge[dashed,bend left=15] node[above] {} ($(s0.south) + (-.33,-.075) $)
    ;

\end{tikzpicture}
    
    \caption{Illustrating configurations for a formula with six variables, where a filled circle denotes an element in the domain, and an unfilled circle an element that is not in the domain. The upper row shows initial $i$-configurations for $i=0,1,2,\ldots, 5 $ (from left to right), the lower row shows looping $i$-configurations for $i=0,1,2,\ldots, 5 $ (from left to right), and the configuration on the right in between the rows is full. Solid arrows show the effect of extending and shifting configurations. The shifting operation is illustrated using the numbers in the circles.
    Finally, in a play of $\game(\tsys,\phi)$ the configurations stored in positions follow the solid arrows, but take the dashed shortcuts avoiding full configurations.}
    \label{fig_playevolutionconcrete}
\end{figure}

We only need certain types of configurations~$c$ for our construction. We say that $c$ is
\begin{itemize}
    \item a full configuration if $\dom{c} = D$, 
    \item an initial $i$-configuration (for $i \in \set{0,1,\ldots, k-1}$) if 
    \[
    \dom{c} =  \set{(j,x) \mid j \in \set{0,1,\ldots,i-1} \text{ and } x \in \set{0,1,\ldots,\lag_j-1} },
    \]
    and
    \item a looping $i$-configuration (again for $i \in \set{0,1,\ldots, k-1}$) if 
    \begin{align*}
        \dom{c} = {}&{} \set{(j,x) \mid j \in \set{0,1,\ldots,i-1} \text{ and } x \in \set{0,1,\ldots,\lag_j-1} } \cup\\
        {}&{} \set{(j,x) \mid j \in \set{i,i+1,\ldots,k-1} \text{ and } x \in \set{0,1,\ldots,\lag_j-2} }.
    \end{align*}
\end{itemize}
Note that for $i = k-1$, both the definition of initial and looping $i$-configuration coincides, as we have $\lag_{k-1} = 1$.
Hence, in the following, we will just speak of $\mathmbox{(k-1)}$-configurations whenever convenient.

Given an initial $i$-configuration~$c$ and a sequence~$\blockseq = \block_0, \block_1, \ldots, \block_{\lag_i-1}$ of $\lag_i$ blocks, we define $\extend(c,\blockseq)$ to be the configuration~$c'$ defined as
\[
c'(j,x) = \begin{cases}
c(j,x) &\text{ if } (j,x) \in \dom{c},\\
\block_x &\text{ if } j = i,\\
\text{undefined} &\text{ otherwise.}
\end{cases}
\]
Furthermore, given a looping $i$-configuration~$c$ and a block~$\block$, we define $\extend(c,\block)$ to be the configuration~$c'$ defined as
\[
c'(j,x) = \begin{cases}
c(j,x) &\text{ if } (j,x) \in \dom{c},\\
\block &\text{ if } j = i \text{ and } x = \lag_i-1,\\
\text{undefined} &\text{ otherwise.}
\end{cases}
\]
Finally, given a full configuration~$c$, we define $\free(c)$ to be the configuration~$c'$ defined as
\[
c'(j,x) = \begin{cases}
c(j,x+1) &\text{ if } x < \lag_j-1,\\
\text{undefined}&\text{ otherwise.}
\end{cases}
\]

The following remark collects how these operations update initial and looping configurations. 

\begin{remark}\hfill
\begin{enumerate}
    \item If $c$ is an initial $i$-configuration for $i < k-1$, then $\extend(c,\blockseq)$ is an initial $\mathmbox{(i+1)}$-configuration.
    \item If $c$ is a $(k-1)$-configuration, then $\extend(c,\block)$ is a full configuration.
    \item If $c$ is a full configuration, then $\free(c)$ is a looping $0$-configuration.
    \item If $c$ is a looping $i$-configuration for $i < k-1$, then $\extend(c,\block)$ is a looping $\mathmbox{(i+1)}$-configuration.
\end{enumerate}
\end{remark}

Finally, let $\autp = (Q, (\pow{\ap})^k, q_\initmark,\delta,\col)$ be a deterministic parity automaton accepting the language of words of the form~$\combine{t_0, t_1, \ldots, t_{k-1}} \in ((\pow{\ap})^k)^\omega$ such that either $t_i \notin \traces(\tsys) $ for some even $i$ or if $\combine{t_0, t_1, \ldots, t_{k-1}} \in L(\aut_\psi^\tsys)$ (i.e., accepting the winning outcomes of plays of $\game(\tsys,\phi)$).

Now, we are able to formally define $\game(\tsys, \varphi)$.
It will be played by the players in $N = \set{1,3,\ldots, k-1}$ (ignoring, for the sake of readability, the fact that $N$ is not of the form~$\set{1,2,\ldots, n}$ for some $n$, as required by the definitions in Subsection~\ref{subsec_distributedgames}). Furthermore, the role of Player~$U$ will be played by Nature.

We define the set of positions to contain all tuples~$(i,c,q)$ where 
$i \in \set{0,1,\ldots, k-1}$, $c$ is an (initial or looping) $i$-configuration, and $q$ is a state of $\autp$, together with a sink state~$\sink$.
The initial position is $(0, c_\bot, q_\initmark)$ where $c_\bot$ is the configuration with empty domain (which is the unique initial $0$-configuration).

We define the action set for Player~$i$ (for odd $i$) as $\blocks^{\lag_i} \cup \blocks$, where actions in $\blocks^{\lag_i}$ are intended for round~$0$ and those in $\blocks$ are intended for all other rounds. If the wrong action is used, then the sink state will be reached.
Next, we define the function~$\owner$ determining which player picks an action to continue a play:
we have $\owner(i,c,q) = i$ for odd $i$, $\owner(i,c,q) = 1$ for even $i$ (we will soon explain how Player~$U$ moves are simulated even though Player~$1$ owns the corresponding positions), and $\owner(\sink) = 1$. 

The set~$E$ of edges is defined as follows (recall that we label edges by actions of a single player, as we define a turn-based game):
\begin{itemize}
    \item We begin by modelling the moves of Player~$U$. Recall that in a turn-based distributed game, Nature resolves the nondeterminism left after the player who is in charge at that positions has picked an action. Thus, we simulate a move of Player~$U$ by giving the position to (say) Player~$1$. Then, we define the edges such that Player~$1$'s move is irrelevant, but the nondeterminism models the choice of Player~$U$.
    
    Formally, for $i \in \set{0,2,\ldots, k-2}$, an initial $i$-configuration~$c$, and a state~$q$ of $\autp$, we have the edge~$((i,c,q),a,(i+1,\extend(c,\blockseq),q))$ for every action~$a$ of Player~$1$ and every $\blockseq \in \blocks^{\lag_i}$: No matter which action~$a$ Player~$1$ picks, Nature can for each possible~$\blockseq$ pick a successor that extends~$c$ by $\blockseq$. This indeed simulates the move of Player~$U$ in subround~$(0,i)$.

    \item For $i \in \set{1,3, \ldots, k-3}$, an initial $i$-configuration~$c$, and a state~$q$ of $\autp$, we have the edge~$((i,c,q),\blockseq,(i+1, \extend(c,\blockseq),q))$ for each $\blockseq \in  \blocks^{\lag_i}$ (this simulates the move of Player~$i$ in subround~$(0,i)$) as well as the edge~$((i,c,q),\block,\sink)$ for each $\block \in  \blocks$ (Player~$i$ may not pick a single block in subround~$(0,i)$ if $i < k-1$).
    Note that there is no nondeterminism to resolve for Nature, as there is a unique successor position for each action.

    \item For a $(k-1)$-configuration~$c$ and state~$q$ of $\autp$, we have the edge~$((\mathmbox{k-1},c,q),\block,(0,\free(c'), q'))$ for each $\block \in \blocks$ (recall that we have $\lag_{k-1} = 1$ here), where $c' = \extend(c,\block)$ and $q'$ is the state reached by $\autp$ when processing $
    \combine{c'(0,0), c'(1,0), \ldots, c'(k-1,0)}$ from $q$.

    \item For $i \in \set{0,2,\ldots, k-2}$, a looping $i$-configuration~$c$, and a state~$q$ of $\autp$, we have the edge~$((i,c,q),a,(i+1,\extend(c,\block),q))$ for every action~$a$ of Player~$1$ and every $\block \in \blocks$: No matter which action~$a$ Player~$1$ picks, Nature can for each possible~$\block$ pick a successor that extends~$c$ by $\block$. This simulates the move of Player~$U$ in a subround~$(r,i+1)$ for $r>0$ using the same mechanism as described above for moves of Player~$U$ in initial configurations.

    \item For $i \in \set{1,3, \ldots, k-3}$, a looping $i$-configuration~$c$, and a state~$q$ of $\autp$, we have the edge~$((i,c,q),\block,(i+1, \extend(c,\block),q))$ for each $\block \in  \blocks$ (this simulates the move of Player~$i$ in subround~$(r,i)$ for $r > 0$) as well as the edge~$((i,c,q),\blockseq,\sink)$ for each $\blockseq \in  \blocks^{\lag_i}$ (Player~$i$ may not pick a sequence of blocks in a subround~$(r,i)$).
    Again, there is no nondeterminism to resolve for Nature, as there is a unique successor position for each action.

    \item For completeness, we have the edge~$(\sink,a,\sink)$ for every action~$a$ of Player~$1$.
\end{itemize}

It remains to define the observation functions~$\beta^i$ as $\beta^i(i,c,q) = \restrict{c}{i}$ where $\restrict{c}{i}$ is the configuration defined as 
\[
(\restrict{c}{i})(j,x) = \begin{cases}
    c(j,x) &\text{ if } j \in \set{0,2, \ldots, i-1},\\
    \text{undefined} &\text{ otherwise,}
\end{cases}
\]
i.e., Player~$i$ can only observe the blocks picked for the variables~$\pi_j$ with even $j < i$.
For completeness, we also define the observation of the sink state~$\sink$ to be $\bot$, where $\bot \notin C$. 
Then, the order~$1 \succeq 3 \succeq \cdots \succeq k-1$ witnesses that the game graph yields hierarchical information.

This completes the definition of the game graph. To complete the definition of the game we define the winning condition as follows:
Let $(i_0,c_0,q_0)(i_1,c_1,q_1)(i_2,c_2,q_2)\cdots$ be a play and let $(i_0,c_0,q_0)(i_k,c_k,q_k)(i_{2k},c_{2k},q_{2k})\cdots$ be the subsequence of all (initial and looping) $0$-configurations.
Recall that the state~$q_{rk}$ in such a position (for $r>0$) is obtained by processing some~$\combine{b_0, \ldots, b_{k-1}}$ from $q_{(r-1)k}$ , where the $b_i$ are blocks picked by the players. 
We say $(i_0,c_0,q_0)(i_1,c_1,q_1)(i_2,c_2,q_2)\cdots$ is in the winning condition, if $\col(q_0)\col(q_k)\col(q_{2k})\cdots$ satisfies the parity condition, which is an $\omega$-regular winning condition. In particular, no winning play may visit the sink~$\sink$.

\begin{remark}
The (concrete) distributed game constructed here is a formalization of the abstract game~$\game(\tsys, \phi)$ described in Section~\ref{sec_abstractgame}. In particular, a winning collection of (finite-state) strategies for the coalition of players in the abstract game corresponds to a winning (finite-state) strategy profile in the concrete game and vice versa.

Hence, whenever convenient below, we do not distinguish between the concrete and the abstract game.
\end{remark}

Now, our main theorem (Theorem~\ref{thm:main}) is a direct consequence of Proposition~\ref{prop_distributedgames} and Lemma~\ref{lemma_correctness}.
Recall that the theorem states that the problem \myquot{Given a transition system~$\tsys$ and a \hyltl sentence~$\phi$ with $\tsys \models \phi$, is $\tsys \models \phi$ witnessed by computable Skolem functions?} is decidable and that, if the answer is yes, our algorithm computes \bddfts implementing such Skolem functions.

\begin{proof}
\label{page:proofofmainthm}
Due to Lemma~\ref{lemma_correctness} and Proposition~\ref{prop_distributedgames}.\ref{prop_distributedgames_finitestate}, the following statements are equivalent:
\begin{itemize}
    \item $\tsys \models\varphi$ has computable Skolem functions.
    \item $\game(\tsys,\varphi)$ has a winning strategy profile.
    \item $\game(\tsys,\varphi)$ has a winning profile of finite-state strategies.
\end{itemize}
The last statement can be decided effectively due to Proposition~\ref{prop_distributedgames}.\ref{prop_distributedgames_decidability}. Thus, the existence of computable Skolem functions is decidable.

Now, assume $\tsys\models\varphi$ has computable Skolem functions. Due to the equivalence above, (the concrete) $\game(\tsys,\varphi)$ has a winning profile~$(s^1, s^3, \ldots, s^{k-1})$ of finite-state strategies. 
We show by induction over~$i$ how the finite-state strategies~$s^i$ can be turned into bounded-delay transducers~$\transd_i$ computing Skolem functions witnessing $\tsys \models \varphi$. The proof follows closely the analogous results shown in Lemma~\ref{lemma_correctness}.\ref{lemma_correctness_fromfsws2compskolem} for Turing machine computable Skolem functions.

Let us fix some $i\in\set{1,3,\ldots, k-1}$ and let $\moore_i$ be a Moore machine for Player~$i$ implementing $s^i$. Recall that $\moore_i$ reads observations of Player~$i$ (configurations of the form~$\restrict{c}{i}$ for initial and looping configurations) and returns actions of Player~$i$. 

On the other hand, $\transd_i$ reads an input~$\combine{t_0, t_2, \ldots, t_{i-1}} \in ((\pow{\ap})^{\frac{i}{2}})^\omega$ where we split each $t_j$ into blocks $t_j = t_j^0t_j^1t_j^2\cdots$.

We construct $\transd_i$ so that it works in two phases, an initialization phase and a looping phase, that is repeated ad infinitum.
We begin by describing the initialization. 
It begins by $\transd_i$ reading $\lag_0$ blocks from each $t_j$.
These blocks can now be assembled into a sequence of observations of Player~$i$ corresponding to the play prefix of $\game(\tsys, \phi)$ in which Player~$U$ picks these blocks. Note that this requires to process each such observation twice, as Player~$i$'s observation does not get updated, when Player~$j$ for odd $j < i$ makes a move.
All unused blocks are stored in the state space of $\transd$. 
Now, $\transd_i$ simulates the run of the Moore machine~$\moore_i$ implementing $s^i$ on this sequence of observations, yielding an action~$a^i$. 
As $s^i$ is winning, this action is a sequence~$\blockseq = b_0, b_1, \ldots, b_{\lag_i-1}$ of blocks.
Then, $\transd_i$ outputs~$b_0b_1\cdots b_{\lag_i-1}$. 
Now, we process the last observation another $k-i$ times with $\moore_i$, simulating the moves for the remaining variables (which are hidden from Player~$i$ which implies the observation is unchanged).
This concludes the initialization phase.

The looping phase begins with $\transd_i$ reading another block from each $t_j$. 
Again, these blocks and the ones stored in the state space can be assembled into observations of Player~$i$ corresponding to a continuation of the simulated play prefix in which Player~$U$ pick these blocks. Then, the run of the Moore machine~$\moore_i$ implementing $s^i$ can be continued, yielding an action~$a^i$.
This is now a block~$b$, which is output by $\transd_i$.
Again, we process the last observation just assembled another $k-i$ times to simulate the moves for the remaining variables, which concludes one looping phase. 
Note that the delay of $\transd_i$ is bounded by~$k\cdot\ell$, where $\ell$ is the block length. 

By storing blocks from the input that have been read but not yet used in observations, by discarding blocks no longer needed, and by keeping track of the state the simulated run of $\moore_i$ ends in, this behaviour can indeed be implemented using a finite state-space. 
We leave the tedious, but straightforward, formal definition of $\transd_i$ to the reader.
\end{proof}

\subsection{Complexity Analysis}
\label{subsec:complexityanalysis}

In this subsection, we determine the complexity of our algorithm. 
Our benchmark here is the complexity of the model-checking problem for \hyltl, which is \tower-complete. More precisely, checking whether a given transition system~$\tsys$ satisfies  a given formula~$\varphi$ is complete for nondeterministic space bounded by a $k$-fold exponential in $\size{\varphi}$, where $k$ is the number of alternations in $\varphi$~\cite{FinkbeinerRS15,Rabe16Diss}.

In the following, we bound the size of $\game(\tsys,\varphi)$ constructed in Subsection~\ref{subsec_concretegame} and then apply known complexity results for solving distributed games, there giving an upper bound on the complexity of deciding whether $\tsys\models\varphi$ is witnessed by computable Skolem functions.

Recall that we start the construction of $\game(\tsys,\varphi)$ with a transition system~$\tsys$ and a \hyltl sentence
$
\phi = \forall \pi_0 \exists \pi_1 \cdots \forall \pi_{k-2} \exists \pi_{k-1}.\ \psi.
$
First, we construct the Büchi automaton~$\aut_\psi^\tsys$ recognizing the language
\[
  \{ \combine{\Pi(\pi_0), \ldots, \Pi(\pi_{k-1})} \mid \Pi(\pi_i) \in \traces(\tsys) \text{ for all $ 0 \le i < k$ and }(\traces(\tsys),\Pi) \models \psi \}.
\]
Due to Remark~\ref{remark_automataconstruction}, we can bound the size of $\aut_\psi^\tsys$ by 
\[
\size{\psi}\cdot 2^{\size{\psi}}\cdot \size{\tsys}^k \le 2^{\log (\size{\psi}\cdot 2^{\size{\psi}}\cdot \size{\tsys}^k)} = 2^{\log (\size{\psi}) + \size{\psi} + k\cdot\log(\size{\tsys})} \le 2^{\bigo(\size{\psi} \cdot \size{\tsys})},
\]
i.e., exponentially in $\size{\psi} \cdot \size{\tsys}$.

Now, let $n$ denote the size of $\aut_\psi^\tsys$ and $s$ the size of $\tsys$. 
Then, $\equivnew{k}$ has index at most~$3^{n^2} \cdot 2^{ks^2} \le 2^{\bigo(n^2 + ks^2)}$: 
There are $3$ choices for each pair~$(p,q)$ of states of $\aut_\psi^\tsys$, i.e., 
\begin{enumerate}
    \item there is a run between $p$ and $q$ that visits at least one accepting state,
    \item there is a run between $p$ and $q$, but none that visits an accepting state, and
    \item there is no run between $p$ and $q$ at all,
\end{enumerate}
resulting in the factor~$3^{n^2}$, and there are two choices for each $j \in \set{0,1,\ldots, k-1}$ and each pair~$(u,v)$ of vertices of $\tsys$, i.e.,
\begin{enumerate}
    \item there is a path of $\tsys$ from $u$ to $v$ and 
    \item there is no path of $\tsys$ from $u$ to $v$,
\end{enumerate}
resulting in the factor~$2^{ks^2}$.

Furthermore, every equivalence class of $\equivnew{k}$, which is a language of finite words over the alphabet~$(\pow{\ap})^k$, is recognized by a DFA with at most
\[(3^n)^n \cdot (2^s)^{s\cdot k}=  3^{n^2}\cdot 2^{ks^2}  \le 2^{\bigo(n^2 + ks^2)}\]
states: 
For each state~$q$ of $\aut_\psi^\tsys$ the DFA keeps track of which states are reachable from $q$ (with and without having seen an accepting state) by processing the input word, which requires the product of $n$ modified powerset automata derived from $\aut_\psi^\tsys$, each with $3^n$ states.
Furthermore, for each component~$j$ and each vertex~$v$ of $\tsys$, it also keeps track of which vertices are reachable from $v$ by paths labeled with the $j$-th component of the input word, again requiring $s$ powerset automata derived from $\tsys$, each with $2^s$ states.
The DFA accepts if exactly the right states and vertices (i.e., those uniquely identifying the equivalence class) are reachable.

Now, let $\equivnew{i+1}$ for $0 < i < k$ have index~$\ind$ and let each of its equivalence classes be accepted by a DFA of size at most~$\autsize$.
Then, the reasoning used to prove Lemma~\ref{lemma_equivifinite} shows that $\equivnew{i}$ has index at most~$2^{\ind}$ and that each equivalence class of $\equivnew{i}$ is recognized by a DFA with at most~$\autsize^\ind$ states, as membership in an $\equivnew{i}$ equivalence class is determined by membership and non-membership in $\equivnew{i+1}$ equivalence classes. This can be checked by the product of $\ind$ many DFA for the equivalence classes (or their complements) of $\equivnew{i+1}$. Note that we are working with DFA, so complementation is for free.

Now, one can show by an induction that the index of $\equivnew{i}$ is bounded by a $(k-i+2)$-fold exponential in $\size{\psi}\cdot\size{\tsys}$.
Using this, a second induction shows that each equivalence class of $\equivnew{i}$ is recognized by a DFA whose size is bounded by a $(k-i+3)$-fold exponential in $\size{\psi}\cdot\size{\tsys}$. 
This relies on the fact that $(2^{x})^y$ is equal to $2^{xy}$, which implies that $\autsize^\ind$ is \myquot{just} exponentially larger than $\autsize$, even though both $\ind$ and $\autsize$ are in general towers of exponentials.

Now, we can bound $\ell$, which is defined so that every~$w$ in a finite equivalence class of $\equivnew{1}$ satisfies~$\size{w} < \ell$. 
As a DFA with $n$ states recognizing a finite language can only accept words of length at most $n-1$, we can bound $\ell$ by the size of the DFA accepting the equivalence classes of $\equivnew{1}$, i.e., by a $(k+2)$-fold exponential in $\size{\psi}\cdot \size{\tsys}$.

The number of vertices of $\game(\tsys,\varphi)$ is $k \cdot \size{C} \cdot \size{\autp}$
where $C$ is the set of configurations and where $\autp$ is a deterministic parity automaton recognizing the language of words of the form~$\combine{t_0, t_1, \ldots, t_{k-1}} \in ((\pow{\ap})^k)^\omega$ such that either $t_i \notin \traces(\tsys) $ for some even $i$ or if $\combine{t_0, t_1, \ldots, t_{k-1}} \in L(\aut_\psi^\tsys)$.

The number of configurations can be bounded by 
\[\size{C} \le \size{\blocks}^{\size{D}} \le \left((\size{\pow{\ap}})^{\ell}\right)^{k^2},\]
i.e., by a $(k+3)$-fold exponential in $\size{\psi}\cdot \size{\tsys}$.
Furthermore, the size of $\autp$ can be bounded doubly-exponentially in $\size{\varphi}\cdot \size{\tsys}$, as it can be constructed as the product of a deterministic parity automaton that is equivalent to the nondeterministic Büchi automaton $\aut_\psi^\tsys$ and $\frac{k}{2}$ deterministic parity automata for the language~$(\pow{\ap})^\omega \setminus \traces(\tsys)$. 
Altogether, the number of vertices of $\game(\tsys,\varphi)$ is bounded by a $(k+3)$-fold exponential in $\size{\psi}\cdot \size{\tsys}$.

As distributed games with hierarchical information can be solved in time that is bounded by an $(n+c)$-fold exponential in the number of states (where $n$ is the number of players and where $c$ is a small constant that depends on the type of winning condition)~\cite{DBLP:conf/focs/PnueliR90}, we conclude that $\game(\tsys,\varphi)$ can be solved in time that is bounded by a 
$(\frac{k}{2}+c)$-fold exponential in the number of vertices of $\game(\tsys,\varphi)$.

Thus, the running time of our algorithm determining whether $\tsys \models \varphi$ is witnessed by computable Skolem functions is bounded by a $(k + \frac{k}{2}+c')$-fold exponential in the size of $\size{\psi}\cdot \size{\tsys}$, where $c'$ is again a small constant.

%%%%%%%%%%%%%%%%%%%%%%%%%%%%%%%%%%%%%%%%%%%%%%%%%%%%%%%%%%%%%%%%%%%%%%
%%%%%%%%%%%%%%%%%%%%%%%%%%%%%%%%%%%%%%%%%%%%%%%%%%%%%%%%%%%%%%%%%%%%%%
%%%%%%%%%%%%%%%%%%%%%%%%%%%%%%%%%%%%%%%%%%%%%%%%%%%%%%%%%%%%%%%%%%%%%%
\section{Related Work}
\label{sec_relatedwork}

As explained in the introduction, computing counterexamples for debugging is the most important application of model-checking. In the framework of \ltl model-checking, a counterexample is a single trace of the system that violates the specification. 
Such a counterexample is typically obtained by running a model-checking algorithm (which are in fact based on searching for such counterexamples). 
Variations of the problem include bounded model-checking~\cite{bmc}, which searches for \myquot{short} counterexamples. 
Also, counterexample-guided abstraction refinement~\cite{cegar} and bounded synthesis~\cite{bosy} rely on counterexample computation.

Counterexamples have not only been studied in the realm of linear-time logics, but also for many other frameworks, e.g., for \forallctl~\cite{46}, for \ctl~\cite{61,77}, for probabilistic temporal logics~\cite{62a,62}, and for discrete-time Markov models~\cite{38}. 

Furthermore, counterexamples have also been studied in the realm of \hyltl model-checking.
Horak et al.~\cite{hypervisana} developed \hypervis, a webtool which provides interactive visualizations of a given model, specification, and counterexample computed by the \hyltl model-checker~\mchyper~\cite{FinkbeinerRS15}.
In complementary work, Coenen et al.~\cite{hyperexp} present a causality-based explanation method for \hyltl counterexamples, again computed by \mchyper.
However, the counterexamples computed by \mchyper are just sets of traces that violate the formula. More specifically, \mchyper only considers the outermost universal quantifiers and returns a variable assignment to those.
This is obviously complete for formulas with quantifier-prefix~$\forall^*$, i.e., without existential quantifiers, but not for more general formulas.
In fact, this approach ignores the dynamic dependencies between universally and existentially quantified variables that is captured by Skolem functions, which we analyze here.
Finally, let us also mention that counterexamples are the foundation of the bounded synthesis algorithm for $\forall^*$\hyltl specifications~\cite{hysy}. 
In this setting, it is again sufficient to only consider sets of traces and not general Skolem functions.

Finally, Beutner and Finkbeiner~\cite{planning} have presented an (incomplete, but sound) reduction from \hyltl model-checking to nondeterministic multi-agent planning so that if the resulting planning problem admits a plan, then the original model-checking problem is satisfied. 
Such plans can be understood as Skolem functions for the existentially quantified variables. 

%%%%%%%%%%%%%%%%%%%%%%%%%%%%%%%%%%%%%%%%%%%%%%%%%%%%%%%%%%%%%%%%%%%%%%
%%%%%%%%%%%%%%%%%%%%%%%%%%%%%%%%%%%%%%%%%%%%%%%%%%%%%%%%%%%%%%%%%%%%%%
%%%%%%%%%%%%%%%%%%%%%%%%%%%%%%%%%%%%%%%%%%%%%%%%%%%%%%%%%%%%%%%%%%%%%%
\section{Conclusion}
\label{sec_conc}

Based on the maxim ``counterexamples/explanations are Skolem functions for the existentially quantified trace variables'', we have shown how to explain why a given transition system satisfies a given \hyltl formula or, equivalently, to provide counterexamples in case the system does not satisfy the formula.
We consider arbitrary computable Skolem functions as explanations.
However, this leads to incompleteness as not every $\tsys \models \varphi$ is witnessed by computable Skolem functions.
Nevertheless, we have shown that the existence of computable explanations is decidable, and that they are effectively computable (whenever they exist).

Recall that the runtime of our algorithm is bounded by a $(k+\frac{k}{2}+c)$-exponential while \hyltl model-checking is complete for nondeterministic space bounded by a $k$-fold exponential in $\size{\varphi}$, where $k$ is the number of alternations in $\varphi$~\cite{FinkbeinerRS15,Rabe16Diss}.
In future work we aim to determine whether computing counterexamples/explanations is inherently harder than model-checking.

After having laid the theoretical foundations, it remains to investigate under which circumstances such explanations can be useful in the verification work-flow. 
For restricted settings, this line of work has been studied as discussed in Section~\ref{sec_relatedwork}. We propose to continue this line of work for more expressive fragments of \hyltl and explanations.
However, for fragments with few quantifier alternations, the situation is clearer. 
For example, we conjecture that results on delay games with LTL winning conditions~\cite{KleinZimmermannLTL} and uniformization of LTL definable relations~\cite{FWjournal} can be adapted to show that deciding the existence of computable Skolem functions for $\forall^*\exists^*$-sentences is \threeexp-complete.

Finally, in future work we compare our work to the \hyltl model-checking algorithm for $\forall^*\exists^*$-sentences of Beutner and Finkbeiner~\cite{BF}, which relies on a delay-free game extended with prophecies that require the player in charge of the universal variables to make commitments about future moves (w.r.t.\ to membership in a finite list of $\omega$-regular properties). In particular, to the best of our knowledge, this approach has not been extended beyond the $\forall^*\exists^*$-fragment.

\subparagraph{Epilogue} After her lunch break, Tracy walks by the printer room and finds a second copy of her document in the printer tray. It turns out the print system satisfies $\phi_{\mathrm{id}}$, but the explanation shows that there is a large delay in the network.

\bibliographystyle{plain}
\bibliography{bib}

\end{document}